\documentclass[10pt,reqno]{amsart}

\newcommand\version{October 14, 2022}

\usepackage[ansinew]{inputenc}

\usepackage{textcomp}

\usepackage{cite}

\usepackage{amsmath,amsfonts,amsthm,amssymb,amsxtra}

\usepackage[usenames,dvipsnames]{color}

\usepackage{bbm}
\usepackage{fancyhdr}
\usepackage[scanall]{psfrag}
\usepackage{graphicx}
\usepackage{ifthen}
\usepackage{pstricks,pst-node,pst-plot}
\usepackage{xcolor}
\usepackage{wasysym}
\usepackage{amsthm}
\usepackage{epstopdf}
\usepackage{breqn}
\usepackage{amssymb}
\usepackage[usenames,dvipsnames]{color}
\usepackage{euscript}
\usepackage{multicol}
\usepackage{amssymb}
\usepackage{amsmath}
\usepackage{graphicx}
\usepackage{caption}
\usepackage{times} 
\usepackage{tikz}
\usepackage{mathrsfs}
\usepackage{float}
\usepackage{enumerate}
\usepackage[linktocpage]{hyperref}
\usepackage{hyperref}
\hypersetup{
    colorlinks=true,
    linkcolor=blue,
    citecolor=red,
    filecolor=magenta,      
    urlcolor=cyan,
    pdftitle={SURGERY TRANSFORMATIONS AND SPECTRAL ESTIMATES OF $\delta$ BEAM OPERATORS},
    pdfpagemode=FullScreen,
    }

\newtheorem{theorem}{Theorem}[section]
\newtheorem{proposition}[theorem]{Proposition}
\newtheorem{lemma}[theorem]{Lemma}
\newtheorem{corollary}[theorem]{Corollary}

\theoremstyle{definition}

\newtheorem{definition}[theorem]{Definition}
\newtheorem{example}[theorem]{Example}

\theoremstyle{remark}

\newtheorem{remark}[theorem]{Remark}

\numberwithin{equation}{section}

\renewcommand{\epsilon}{\varepsilon}

\begin{document}

\title[Delta Beam Operators--- \version]{Surgery transformations and spectral estimates of $\delta$ beam operators}

\author{Aftab Ali and Muhammad Usman}
\address{Aftab Ali, Department of Mathematics, Lahore University of Management Sciences (LUMS), Lahore}
\email{18070008@lums.edu.pk}
\address{Muhammad Usman, Department of Mathematics, Lahore University of Management Sciences (LUMS), Lahore}
\email{m.usman08@alumni.imperial.ac.uk}


\begin{abstract} 
We introduce $\delta$ type vertex conditions for beam operators, the fourth derivative operator, on metric graphs and study the effect of certain geometrical alterations (graph surgery) of the graph on the spectra of beam operators on compact metric graphs. Results are obtained for a class of vertex conditions which can be seen as an analogue of $\delta$ vertex conditions for quantum graphs. There are a number of possible candidates of $\delta$ type conditions for beam operators. We develop surgery principles and record the monotonicity properties of the spectrum, keeping in view the possibility that vertex conditions may change within the same class after certain graph alterations. We also demonstrate the applications of surgery principles by obtaining several lower and upper estimates on the eigenvalues.      
\end{abstract}

\maketitle
\tableofcontents
\section{Introduction}
The study of spectral estimates of quantum graphs has attracted significant attention in the last couple of decades. One of the major objectives of these studies is to shed light on the relationships between spectrum and topological properties of quantum networks and, by doing so, develop methods and techniques to estimate the eigenvalues in terms of parameters associated with differential operators and graph topology or connectivity. Most of the recent studies have been focused on investigating the effect of surgical operations on the spectrum of a given graph and how can one use a set of surgical operations for spectral estimations by transforming a graph into a simpler one by precisely predicting the corresponding change in spectrum. This leads to certain spectral estimates on the original graph in terms of quantities such as total length, Betti number and number of edges and vertices. These techniques were applied successfully to derive lower and upper bounds on the spectral gap and higher eigenvalues for graph Laplacian equipped with Kirchhoff, $\delta$ and Dirichlet vertex conditions, see \cite{BKKM} and references therein. \\

In the current paper, motivated by \cite{BKKM}  and \cite{KM}, we discuss the surgery principles and spectral estimates of fourth order operator $\frac{d^4}{dx^4}$ on metric graphs equipped with $\delta-$type conditions. We will refer to this operator as $\delta$ beam operator. 
Several possible vertex conditions for beam operators correspond to analogous $\delta$-type conditions for Laplacian. Among all others, we only consider four of their kinds, see vertex condition \eqref{vc1}, \eqref{gvc1},  \eqref{gvc2} and \eqref{vc2}. The special cases of these conditions have already appeared at several places in connection with the studies of Euler-Bernoulli beam operators. For instance, in \cite{DN} Dekoninck and Nicaise study the characteristic equation for the spectrum of the Euler-Bernoulli beam operator on networks equipped with vertex conditions \eqref{vc1}  and \eqref{gvc2} with zero strengths and showed the dependence of spectrum  on the structure of the graph. In \cite{DN1} they study the exact controllability problem of hyperbolic systems of networks of beams subject to the special cases of vertex conditions \eqref{gvc1} with zero strengths. Gregorio and Mungnolo in \cite{GM} have discussed the qualitative properties of one-dimensional parabolic
evolution equations associated with the fourth order differential operator, defined on metric graphs equipped with these vertex conditions, and have investigated how different transmission conditions in the vertices of a network can arouse rather different behaviors of solutions to such partial differential equations. In \cite{KKU}, based on physical-geometrical considerations of beams embedded in $\mathbb{R}^2$, self-adjoint vertex conditions were derived, which take into account the geometry of the underlying graph and depend on angles between the edges. The obtained vertex conditions is a particular case of \eqref{gvc2}. In three dimension, the vertex conditions and corresponding self-adjoint beam operators for rigid joints of elastic beam frames have recently appeared in \cite{SBE, GBE}.\\

Surgery principles for beam operators were recently studied by Kurasov and Muller in \cite{KM}. Vertex condition \eqref{vc2} with zero strengths were considered and certain lower and upper estimates were obtained on the spectral gap. For Laplacian, on the other hand, there has been significant progress on the topic of spectral estimation using surgery principles and other methods, see \cite{A, BKKM1,KKK,KKMM,F,R,RS, KN, KMN}.  In \cite{N} Nicaise showed that among all graphs with a fixed length, the minimal spectral gap is obtained for the single-edge graph.
In \cite{F} Friedlander proved a more general result, showing that the minimum of the $k$-th eigenvalue is uniquely obtained for a star graph with $k$ edges. Exner and Jex \cite{EJ} showed how the change of graph edge lengths may increase or decrease the spectral gap, depending on the graph's topology. In the last couple of years, a series of works on the subject came to light. Kurasov and Naboko \cite{PK} treated the spectral gap minimization and together with Malenov\'a they explored how spectral gap changes with various modifications of graph connectivity \cite{KMN}. Kennedy, Kurasov, Malenov\'a and Mugnolo provided a broad survey on bounding the spectral gap in terms of various geometric quantities of the graph \cite{KKMM}. Karreskog, Kurasov and Kupersmidt generalized the minimization results mentioned above to Schr\"odinger operators with potentials and $\delta$-type vertex conditions. Del Pezzo and Rossi proved upper and lower bounds for the spectral gap of the $p$-Laplacian and evaluated its derivatives with respect to change of edge lengths
\cite{PR}. Rohleder solved the spectral gap maximization problem for all eigenvalues
of tree graphs \cite{R}. Ariturk provides some improved upper bounds for all graph eigenvalues \cite{A}. Berkolaiko, Kennedy, Kurasov and
Mugnolo further generalize lower and upper bounds of the spectral gap in terms of the edge connectivity \cite{BKKM1} and intensively studied the surgery principles for Laplacian in \cite{BKKM}.
\\

The paper is organized as follows. In Section $2$ we introduce the so called $\delta$ beam operators. The quadratic forms and their domains associated with different vertex conditions are discussed.  This is followed by some preliminary results in Section $3$. Section $4$ deals with the eigenvalues dependence on the vertex conditions, these results co-relate the eigenvalues of same underlying metric graph subject to different vertex conditions. In particular, \textcolor{blue}{Theorem} \eqref{strength} shows that the eigenvalues depend positively on the vertex interaction strengths. Section $5$ and $6$ are regarding the surgery principles for the metric graphs with beam operators. In Section $5$, we consider surgical transformations which changes  graph's connectivity but total length remains the same. These operations include the gluing and splitting of vertices.  This surgical operation is classified into two major parts depending on vertex conditions. Depending on the choice of vertex conditions, the eigenvalues of glued graph are interlaced by degrees one or two with the eigenvalues of unglued graph. 
The operation of splitting a vertex, in general, does not behave as a converse of gluing. Section $6$ deals with the surgical operation that increases the total length of a graph. This includes attaching a pendant graph and inserting a graph at some vertex of another graph. Both of these operations lower the eigenvalues of a metric graph. Section $7$  and $8$ are devoted to applying surgery principles to obtain the eigenvalues estimates. Some of the estimates on the lowest non-zero eigenvalue are obtained using some trial functions from the domain of the quadratic form.  The bounds on general eigenvalues are derived by using topological alterations of graphs affecting the eigenvalues in a predictable manner; with these techniques, certain classes of metric graphs permit stronger upper estimates (bipartite graphs), whilst others permit stronger lower estimates (Eulerian graphs). To obtain improved spectral estimates for beam operators appearing in \cite{SBE, GBE}, we plan to extend the surgery principles to such formulation of beam operators in our future work.
\section{The $\delta$ beam operators}
Let $\Gamma$ be a finite compact and connected metric graph with edge set $E=\{e_i\}_{i=1}^{|E|}$ and vertex set $V=\{v_j\}_{j=1}^{|V|}$. Each edge $e_i$ is identified with an interval $[x_{2i-1}, x_{2i}]$ of the real-line and each vertex $v_j$ can be considered as a partition of the set of all end points $\{x_k\}_{k=1}^{2|E|}$.\\

We define the \textit{$\delta$ beam operators} via the closed quadratic form
\begin{equation}\label{QF}
h[\varphi]:=\int_\Gamma |\varphi''(x)|^2\,dx+\sum_{m=1}^{|V|}\alpha_m|\varphi(v_m)|^2
\end{equation}
with $\varphi\in W^2_2(\Gamma\backslash V)$ and $\alpha_m\in\mathbb{R}$. It is natural to assume that the beams are connected at vertices. This means that the functions satisfy continuity condition 
\begin{equation}\label{continuity}
\varphi(x_i)=\varphi(x_j)\equiv\varphi(v_m), \quad x_i, x_j \in v_m
\end{equation}
at the vertices.
If no further conditions are assumed on the domain of the quadratic form $h$ then \eqref{QF} corresponds to the fourth order differential operator 
$\frac{d^4}{dx^4}$
 defined on the functions $\varphi\in W_2^4(\Gamma \backslash V)$ satisfying vertex conditions
\begin{equation}
\tag{C1}
\label{vc1} \begin{cases} \varphi(x_i)=\varphi(x_j)\equiv\varphi(v_m),\quad x_i,x_j\in v_m,\\ \partial^2\varphi(x_j)=0, \\
\sum\limits_{x_j\in v_m}\partial^3\varphi(x_j)=-\alpha_m\varphi(v_m).
\end{cases}\end{equation}
In addition to the continuity condition \eqref{continuity} if we further assume that normal derivatives of functions from the domain of the quadratic form $h$ satisfies the condition
\begin{equation}\label{gfirstderivative}
\sum_{x_j\in v_m}\overline{ \sigma^{v_m}_{x_j}}\partial\varphi(x_j)=0, \quad \sigma^{v_m}_{x_j}\in \mathbb{C}\backslash \{0\},
\end{equation}
then, the corresponding self-adjoint fourth order differential operator
$ 
\frac{d^4}{dx^4}$
is defined on functions from $W_2^4(\Gamma\backslash V)$ satisfying the following conditions at vertices
\begin{equation}
\tag{C2}
\label{gvc1} \begin{cases} \varphi(x_i)=\varphi(x_j)\equiv\varphi(v_m),\quad x_i,x_j\in v_m,\\ \sum\limits_{x_j\in v_m}\overline{\sigma^{v_m}_{x_j}}\partial\varphi(x_j)=0, \quad \sigma^{v_m}_{x_j}\in \mathbb{C}\backslash \{0\}, \\
\frac{\partial^2\varphi(x_i)}{\sigma^{v_m}_{x_i}}=\frac{\partial^2\varphi(x_j)}{\sigma^{v_m}_{x_j}},\\
\sum\limits_{x_j\in v_m}\partial^3\varphi(x_j)=-\alpha_m\varphi(v_m).
\end{cases}\end{equation}
If $\sigma^{v_m}_{x_j}=0$, we set $\partial^2\varphi(x_j)=0$ and impose no further condition on $\partial\varphi(x_j)$. Vertex conditions \eqref{vc1} corresponds to the case when  $\sigma^{v_m}_{x_j}$ are equal to zero in \eqref{gvc1}.\\
\begin{remark} \label{remarkk}
 \noindent For condition \eqref{gvc1} with $\alpha_m=0, m=1,2,\cdots,|V|$,
 and $\sigma^{v_m}_{x_i}=\sigma^{v_m}_{x_j}$, any interior point of an edge can be considered as a vertex $v_m$ of degree two equipped with \eqref{gvc1}  with $\alpha_m=0$. Since at any interior point functions and their first, second and third derivatives are continuous. Similarly, any vertex $v_m$ of degree two equipped with vertex condition \eqref{gvc1} with $\alpha_m=0$ can be regarded as an interior point and two edges are replaced by a single long edge. 
\end{remark}


A special set of vertex conditions of the form \eqref{gvc1} for planar graphs was discussed in \cite{KKU}. The angle dependent vertex conditions were obtained by assuming that functions from the domain of the quadratic form \eqref{QF} are continuous at vertices and that the tangential lines to the graph of functions at a vertex lie in the same plane. For instance,  if we consider any point $v$ on an edge $e$ as a degree two vertex, that is $v=\{x_1,x_2\}$,  then the continuity and  coplaner condition for tangent lines are given by
\begin{equation}\label{degreetwo}
\begin{cases}\varphi(x_1)=\varphi(x_2)=\varphi(v), \\
 \partial\varphi(x_1)+\partial\varphi(x_2)=0.\\
\end{cases}
\end{equation}
The corresponding self-adjoint operator $\frac{d^4}{dx^4}$ defined on $W_2^4(E\backslash \{v\})$ satisfying the conditions \eqref{degreetwo} and
$$
\begin{cases}\partial^2\varphi(x_1)=\partial^2\varphi(x_2),\\
\partial^3\varphi(x_1)+\partial^3\varphi(x_2)=-\alpha\varphi(v).\\
\end{cases}
$$
Generally, if $v_m$ is a vertex of degree $d_m$ such that angles between edges are $\{\gamma_j\}_{j=1}^{d_m}$, none of which equal to $0$ or $\pi$  then we have the following angle dependent vertex conditions
\begin{equation}\label{angleconditions}
\begin{cases} \varphi(x_i)=\varphi(x_j)\equiv\varphi(v_m),\quad x_i,x_j\in v_m,\\ \sum\limits_{x_j\in v_m}\sin(\gamma_j)\cdot \partial\varphi(x_j)=0, \\
\frac{\partial^2\varphi(x_i)}{\sin(\gamma_i)}=\frac{\partial^2\varphi(x_j)}{\sin(\gamma_j)},\\
\sum\limits_{x_j\in v_m}\partial^3\varphi(x_j)=-\alpha_m\varphi(v_m).
\end{cases}
\end{equation}
The conditions when one of the angles is equal to $0$ or $\pi$ are given in Section 4 of \cite{KKU}.\\
Note that,  changing the role of the first and second normal derivatives in the vertex conditions \eqref{gvc1} does not effect the self-adjointness of the operator. This allows us to define the \textit{$\delta$ beam operator} given by the same differential expression $\frac{d^4}{d x^4}$ on the functions from $W_2^4(\Gamma\backslash V)$ satisfying the vertex conditions 
\begin{equation}
\tag{C3}
\label{gvc2} \begin{cases} \varphi(x_i)=\varphi(x_j)\equiv\varphi(v_m),\quad x_i,x_j\in v_m,\\
\frac{\partial\varphi(x_i)}{\sigma^{v_m}_{x_i}}=\frac{\partial\varphi(x_j)}{\sigma^{v_m}_{x_j}}, \quad \sigma^{v_m}_{x_j}\in \mathbb{C}\backslash \{0\},
\\
\sum\limits_{x_j\in v_m}\overline{\sigma^{v_m}_{x_j}}\partial^2\varphi(x_j)=0,\\
\sum\limits_{x_j\in v_m}\partial^3\varphi(x_j)=-\alpha_m\varphi(v_m).
\end{cases}\end{equation}
If $\sigma^{v_m}_{x_j}=0$, we set  $\partial\varphi(x_j)=0$ and do not assume any condition on $\partial^2\varphi(x_j)$. In this case the vertex conditions \eqref{gvc2} become
\begin{equation}
\tag{C4}
\label{vc2} \begin{cases} \varphi(x_i)=\varphi(x_j)\equiv\varphi(v_m),\quad x_i,x_j\in v_m,\\ \partial\varphi(x_j)=0, \\
\sum\limits_{x_j\in v_m}\partial^3\varphi(x_j)=-\alpha_m\varphi(v_m).
\end{cases}\end{equation}
\\

When $\alpha_0=\infty$ at certain vertex $v_0$ then functions vanish at all endpoints of intervals incident to $v_0$. We shall call vertex conditions \eqref{vc1}-\eqref{vc2} when $\alpha_0=\infty$ the extended vertex conditions. In this case, the continuity condition is replaced by $\varphi(x_j)=0$, and the last equation in all four conditions becomes redundant. Note that the assumption $\alpha_0=\infty$ in \eqref{vc1} and \eqref{vc2} has the effect of disconnecting edges at vertex $v_0$, and the spectrum of a graph is equal to the union of the spectrum of disjoints intervals. However, we can not decouple a graph into disjoints intervals for conditions \eqref{gvc1} and \eqref{gvc2}. Because functions living on incident edges of a vertex are connected at vertices due to the vanishing of a scalar combination of normalized derivative in \eqref{gvc1} and normalized continuity on the first derivative in \eqref{gvc2}. For a boundary vertex, vertex of degree one, the conditions \eqref{gvc1} and \eqref{vc2} are equivalent, and the conditions \eqref{vc1} and \eqref{gvc2} also coincides. \\

\section{Preliminary results}
A standard method in obtaining eigenvalue estimates is based on variational arguments comparing the Rayleigh quotients on suitable finite dimensional subspaces. In this Section we collect results, without proofs, that will be needed in subsequent sections. We will follow the strategy of \cite{BKKM}, where Laplacian is considered on compact graphs with $\delta$ and Dirichlet vertex conditions. We will need the following minimax characterisation of eigenvalues. 
 \begin{proposition}\label{minmax}
Let $H$ be a self-adjoint operator on metric graph $\Gamma$ with discrete spectrum and $h$ be its quadratic  form which is semi-bounded from below. Then, the eigenvalues $\lambda_k(\Gamma) $ satisfy
\begin{align}\label{minmax1}
\lambda_k(\Gamma) & = \min_{\underset{\dim(X_k) =k }{X_k \subset D\left(h\right) }} M(h,X_k) \\
&= \max_{\underset{\dim(X_{k-1}) =k-1 }{X_{k-1} \subset D\left(h\right) }} m(h,X_{k-1}^{\perp})
\label{minmax2}
\end{align}  

where, $M(h,X_k)=\underset{0\neq\varphi\in X_k}{ \mathrm{max}} \frac{h[\varphi]}{||\varphi||^2}$ and $m(h,X_k)=\underset{0\neq\varphi\in X_{k}}{ \mathrm{min}} \frac{h[\varphi]}{||\varphi||^2}$.
\end{proposition}

If a $k$-dimensional subspace realises the minimum $\lambda_k(\Gamma)$ in \eqref{minmax1} we call it a minimising subspace for $\lambda_k(\Gamma)$.
The following Lemma characterizes the equality.  
\begin{lemma}\cite[\textcolor{red}{Lemma 4.1}]{BKKM} \label{lem}. Any minimising subspace $X_k$ for eigenvalue $\lambda_k(\Gamma)$ contains an eigenvector of $\lambda_k(\Gamma)$. Moreover, if $\lambda_k(\Gamma)<\lambda_{k+1}(\Gamma)$, then the eigenspace of $\lambda_k(\Gamma)$ is the   intersection of its all possible $k$-dimensional minimising subspaces.
\end{lemma}
Let $H$ and $\tilde{H}$ be two self-adjoint operators with discrete spectrum and  $h$ and $\tilde{h}$ be their corresponding semi-bounded from below quadratic forms. We say that $\tilde{h}$ is a \textit{positive rank-n perturbation} of $h$  if $\tilde{h}=h$ on some $Y\subset_n \mathrm{dom}(h)$  and either  $Y=\mathrm{dom}(\tilde{h})\subset_n \mathrm{dom(h)}$ or $\tilde{h}\geq h$ with $\mathrm{dom}(\tilde{h})=\mathrm{dom}(h)$. Here, the symbol $\subset_n$ denotes the subspace such that the quotient space $\mathrm{dom}(h)/Y$ is $n$-dimensional. 

\begin{theorem}\cite[\textcolor{red}{Theorem 4.3}]{BKKM}\label{berk1}. Let $H$ and $\tilde{H}$ be two self-adjoint operator with discrete spectrum defined on metric graphs $\Gamma$ and $\tilde{\Gamma}$, respectively. Let $h$ and $\tilde{h}$ be their corresponding semi-bounded from below quadratic forms. If form $\tilde{h}$ is a positive rank-1 perturbation of the form $h$ then, eigenvalues of $H$ and $\tilde{H}$ satisfy
$$
\lambda_k(\Gamma)\leq \lambda_k(\tilde{\Gamma})\leq \lambda_{k+1}(\Gamma)\leq \lambda_{k+1}(\tilde{\Gamma}), \quad k\geq 1.
$$ 
\end{theorem}
\section{Eigenvalues dependence on vertex conditions}
In this section, we present the dependence of eigenvalues of a quantum graph on vertex conditions. Even a small perturbation in the parameters of a vertex condition can shift the spectrum in some direction. We describe the relationship between eigenvalues of graphs equipped with either of four conditions. The \textcolor{blue}{Theorem} \eqref{strength} elucidates how altering a vertex condition at a single vertex affects the whole spectrum. The general monotonicity principle for eigenvalues holds when, at a certain vertex, the interaction strength is increased or at that vertex, the condition is replaced by some other vertex condition. \\ \\
Let $\frac{d^4}{dx^4}$ be the fourth-order differential operator acting on edges of finite compact metric graphs $\Gamma^i$ for $i=1,2,3,4$, and these graphs $\Gamma^i$ are identical and are equipped with same self-adjoint vertex conditions at all vertices except at $v_0$, where we impose vertex conditions \eqref{vc1}, \eqref{gvc1}, \eqref{gvc2}, and \eqref{vc2}, respectively, for $i=1,2,3,4$. Let $h^1$, $h^2$, $h^3$ and $h^4$ be the corresponding quadratic forms. In order to understand how vertex conditions influence the spectrum of a quantum graph, we need to describe the relationships between eigenvalues of each graph $\Gamma^i$. To this end, a standard method for obtaining eigenvalue estimates is based on variational arguments. In simple terms, the variational argument is based on comparing quadratic forms on suitable finite-dimensional subspaces, and the desired estimates can be obtained with the help of the min-max principle. However, the quadratic form for each graph $\Gamma^i$ has the same expression; therefore, we need to describe the domain of quadratic forms and compare the respective domains. Thus, for functions $\varphi$ from $W^2_2(\Gamma^1 \backslash V)$, which are also continuous on the whole graph, we have 
$$ D(h^1)=\{\varphi\in W^2_2(\Gamma^1\backslash V)   : \varphi \text{ is continuous on $\Gamma^1$}   \},\quad D(h^2)=\{\varphi \in  D(h^1 ): \sum\limits_{x_j\in v_m}\overline{ \sigma^{v_m}_{x_j}}\partial\varphi(x_j)=0\}, $$
$$ D(h^3)=\left \{\varphi \in  D(h^1 ): \frac{\partial\varphi(x_i)}{\sigma^{v_m}_{x_i}}=\frac{\partial\varphi(x_j)}{\sigma^{v_m}_{x_j}} \right \}, \quad D(h^4)=\{\varphi \in  D(h^1 ): \partial \varphi(x_j)=0, x_j \in v_0 \}.$$
The following two inclusion for the respective form domains holds.
\begin{equation} \label{inclusions}
    D(h^4) \subset  D(h^2) \subset  D(h^1)  \quad 
\text{and} \quad 
 D(h^4) \subset  D(h^3) \subset  D(h^1).
\end{equation}
 
{Since quadratic forms $h^i$ and  $h^j$ agree on $D(h^4)$.
 Therefore, minimizing over the smaller space makes the eigenvalue large. Thus 
 \begin{equation} \label{inequalities}
     \lambda_k(\Gamma^1) \leq \lambda_k(\Gamma^2) \leq \lambda_k(\Gamma^4)  \quad \text{and} \quad
\lambda_k(\Gamma^1) \leq \lambda_k(\Gamma^3) \leq \lambda_k(\Gamma^4) 
 \end{equation}
 Moreover, the domains $D(h^2)$, $D(h^3)$ and $D(h^4)$ are co-dimension one, $d-1$ and $d$  subspaces of $D(h^1)$, respectively, where $d$ is the degree of vertex $v_0$. Thus, the quadratic forms' monotonicity and the rank-$n$ nature of the perturbation along with \eqref{inclusions} leads to the following inequalities.
 \begin{equation} 
     \lambda_k(\Gamma^1) \leq \lambda_k(\Gamma^2) \leq \lambda_{k+1}(\Gamma^1), \quad
\lambda_k(\Gamma^1) \leq \lambda_k(\Gamma^4) \leq \lambda_{k+d}(\Gamma^1), \quad \lambda_k(\Gamma^1) \leq \lambda_k(\Gamma^3) \leq \lambda_{k+d-1}(\Gamma^1). 
 \end{equation}
The above set of inequalities also describes a relationship between eigenvalues of $\Gamma^2$ and $\Gamma^3$. The above discussion is useful and helps us prove the following proposition which predicts the influence of vertex conditions on the spectrum of a graph. We found that there is exactly one eigenvalue of the graph $\Gamma^1$ between any two consecutive eigenvalues of its rank one perturbations.}
\begin{proposition} \label{changeofcondition}
\begin{enumerate}
    {\item Let $\Gamma$ be a finite compact metric graph and arbitrary self-adjoint vertex conditions are imposed at each vertex of $\Gamma$, except $v_0$, where conditions (\ref{vc1}) are imposed. Let   $\tilde{\Gamma}$ be a metric graph obtained from $\Gamma$ by imposing vertex conditions \eqref{gvc1},  (preserving the strengths $\alpha_0$) at vertex $v_0$, then
\begin{equation*}
    \lambda_k(\Gamma) \leq \lambda_k (\tilde{\Gamma}) \leq \lambda_{k+1}(\Gamma ).
\end{equation*}
Furthermore, if the graph $\Gamma$ has $|V|$ number of vertices, and vertex conditions at each vertex are changed from \eqref{vc1} to \eqref{gvc1}, then
$$\lambda_k(\Gamma) \leq \lambda_k (\tilde{\Gamma}) \leq \lambda_{k+|V|}(\Gamma ).$$
\item  If the vertex conditions (\ref{vc1}) are imposed at $v_0$ in $\Gamma$, and let $\tilde{\Gamma}$ be graph obtained from $\Gamma$ by imposing vertex conditions (\ref{gvc2}) at $v_0$. Then
 \begin{equation*}
     \lambda_k(\Gamma) \leq \lambda_k (\tilde{\Gamma}) \leq \lambda_{k+d-1}({\Gamma}).
 \end{equation*}
 \item If the vertex conditions (\ref{vc1}) are imposed at $v_0$ in $\Gamma$, and let $\tilde{\Gamma}$ be graph obtained from $\Gamma$ by imposing vertex conditions (\ref{vc2}) at $v_0$. Then
 \begin{equation*}
     \lambda_k(\Gamma) \leq \lambda_k (\tilde{\Gamma}) \leq \lambda_{k+d}(\Gamma).
 \end{equation*}
 \item If the vertex conditions (\ref{gvc1}) are imposed at $v_0$ in $\Gamma$, and let $\tilde{\Gamma}$ be graph obtained from $\Gamma$ by imposing vertex conditions \eqref{gvc2} at $v_0$. Then
 \begin{equation*}
     \lambda_{k}( \Gamma) \leq \lambda_{k+1} (\tilde{\Gamma}) \leq \lambda_{k+d}(\Gamma).
 \end{equation*}
 \item If the vertex conditions (\ref{gvc1}) are imposed at $v_0$ in $\Gamma$, and let $\tilde{\Gamma}$ be graph obtained from $\Gamma$ by imposing vertex conditions \eqref{vc2} at $v_0$. Then
 \begin{equation*}
     \lambda_{k}( \Gamma) \leq \lambda_{k} (\tilde{\Gamma}) \leq \lambda_{k+d-1}(\Gamma).
 \end{equation*}
 \item If the vertex conditions (\ref{gvc2}) are imposed at $v_0$ in $\Gamma$, and let $\tilde{\Gamma}$ be graph obtained from $\Gamma$ by imposing vertex conditions \eqref{vc2} at $v_0$. Then
 \begin{equation*}
     \lambda_{k}( \Gamma) \leq \lambda_{k} (\tilde{\Gamma}) \leq \lambda_{k+1}(\Gamma).
 \end{equation*}}
\end{enumerate}
\end{proposition}
\begin{proof}
\begin{enumerate}
    \item Changing the vertex conditions from (\ref{vc1}) to either (\ref{gvc1}),(\ref{gvc2}), or (\ref{vc2}) cut down on domain of the quadratic form. Since,  $D(\tilde{h})$ is a co-dimension one subspace of $D(h)$. Rest of proof follows from rank one nature of the perturbation. The other parts of the proof can be proved similarly.
\end{enumerate}
\end{proof}
The above interlacing inequalities are very useful in obtaining estimates on eigenvalues, especially when comparing eigenvalues of the same underlying graph equipped with different vertex conditions. If the spectrum of one of the graphs is known, then the spectrum of another graph can be estimated accordingly. To have a clearer view of its application, we provide the following example to illustrate the above inequality.
\begin{example}
Let $\mathcal{C}$ be a loop of length $\ell$, parameterized by $[0,\ell]$, and the vertex $v_0$ is equipped with condition \eqref{gvc1}  with $\sigma_0=\sigma_{\ell}$, or \eqref{gvc2} with $\sigma_0=-\sigma_{\ell}$, with strength zero, then the non-zero eigenvalues are $\lambda_{k}(\mathcal{C})=\left(\frac{ 2 k \pi}{\ell} \right)^4$, each having multiplicity two.  Let $\mathcal{C}'$ be the loop obtained from $\mathcal{C}$ by imposing  vertex conditions  \eqref{vc2} at $v_0$ with $\alpha_0=0$, then $\left(\frac{2 \pi}{\ell} \right)^4$ is the first non-zero eigenvalue. Moreover, other eigenvalues are zeros of the following equation.
\begin{equation} \label{secular}
    \sin\left (\sqrt[4]{\lambda} \ell \right) \left(1-\cosh \left({\sqrt[4]{\lambda}}\ell \right) \right)+\sinh \left ({\sqrt[4]{\lambda}} \ell \right ) \left(1-\cos \left(\sqrt[4]{\lambda} \ell \right) \right)=0.
\end{equation}
Furthermore, one can obtain a subsequence of eigenvalues of $\mathcal{C}'$ by using \textcolor{blue}{Proposition} \eqref{changeofcondition}. Since eigenvalues of $\mathcal{C}$ are 
$$0, \left(\frac{2 \pi}{\ell} \right)^4,\left(\frac{2 \pi}{\ell} \right)^4, \left(\frac{4 \pi}{\ell} \right)^4,\left(\frac{4 \pi}{\ell} \right)^4,\left(\frac{6 \pi}{\ell} \right)^4,\left(\frac{6 \pi}{\ell} \right)^4, \cdots$$
and $\lambda_k(\mathcal{C}) \leq \lambda_k(\mathcal{C}') \leq \lambda_{k+1}(\mathcal{C})$. Therefore, 
\begin{align*}
    \left(\frac{2 \pi}{\ell} \right)^4 &\leq \lambda_1(\mathcal{C}') \leq \left(\frac{2 \pi}{\ell} \right)^4,\quad
    \left(\frac{2 \pi}{\ell} \right)^4 \leq \lambda_2(\mathcal{C}') \leq \left(\frac{4 \pi}{\ell} \right)^4, \quad \left(\frac{4 \pi}{\ell} \right)^4 \leq \lambda_3(\mathcal{C}') \leq \left(\frac{4 \pi}{\ell} \right)^4,\\
    \left(\frac{4 \pi}{\ell} \right)^4 &\leq \lambda_4(\mathcal{C}') \leq \left(\frac{6 \pi}{\ell} \right)^4,\quad
    \left(\frac{6 \pi}{\ell} \right)^4 \leq \lambda_5(\mathcal{C}') \leq \left(\frac{6 \pi}{\ell} \right)^4, \quad \left(\frac{6 \pi}{\ell} \right)^4 \leq \lambda_6(\mathcal{C}') \leq \left(\frac{8 \pi}{\ell} \right)^4.
\end{align*}
\end{example}
It is evident from \cite{BKKM, BK} that, for Laplacian, the operation of changing vertex conditions by varying the $\delta$ potential at a vertex affects the spectrum of the graph. Following their ideas, we show that a similar result holds for the fourth-order differential operator.
\begin{theorem} \label{strength}
Let $H$, and $\tilde{H}$ be two self-adjoint differential operators given by the differential expression $\frac{d^4}{dx^4}$ sharing the same underlying, compact and finite metric graph $\Gamma$ and the same type of delta interactions at all vertices. Assume $\alpha_m$ and $\tilde{\alpha}_m$ are the strengths of delta interactions at vertex $v_m$ associated with the operators $H$ and $\tilde{H}$, respectively. If
\begin{equation}\label{strengthineq}
\alpha_m < \tilde{\alpha}_m,\quad m=1,\cdots,|V|,
\end{equation}
then the conclusions of \textcolor{blue}{Theorem} \eqref{berk1} hold for the eigenvalues of operators $H$ and $\tilde{H}$.\\
Furthermore, if the $k$-th eigenvalue of $\Gamma_ {\tilde{\alpha}}$ is simple and the corresponding eigenfunction is non-zero at $v$, then
$$\lambda_k(\Gamma_{\alpha}) < \lambda_k(\Gamma_{\tilde{\alpha}}).$$
\end{theorem}
\begin{proof}
Let $h_\alpha$ and $h_{\tilde{\alpha}}$ denote quadratic forms of operator $H$ and $\tilde H$, respectively. As both operators share the same underlying metric graph and type of delta interactions at all vertices, therefore,  the domains, $D(h_\alpha)$ and $D(h_{\tilde{\alpha}})$, of both quadratic forms, are equal. Inequality \eqref{strengthineq} implies that quadratic forms $h_{\alpha}$ and $ h_{\tilde{\alpha}}$ satisfy
$$
h_\alpha[\varphi]=\int_\Gamma |\varphi''(x)|^2\,dx+\sum_{m=1}^{|V|} \alpha_m|\varphi(v_m)|^2\leq \int_\Gamma |\varphi''(x)|^2\,dx+\sum_{m=1}^{|V|} \tilde \alpha_m|\varphi(v_m)|^2= h_{\tilde{\alpha}}[\varphi].
$$
Moreover, 
$$
h_{\tilde{\alpha}}=h_\alpha \quad \mbox{if} \quad \varphi(v_m)=0,\quad m=1,\cdots,|V|.
$$
This implies that $h_{\tilde{\alpha}}$ is a positive rank-1 perturbation of $h_{\alpha}$ and hence conclusions of \textcolor{blue}{Theorem} \eqref{berk1} hold for the eigenvalues of operators $H$ and $\tilde H$.

{
To prove that inequality is strict, assume that corresponding to each first $k$ eigenvalues of a graph $\Gamma_{\tilde{\alpha}}$, there is a $k$-dimensional subspace $X_k$ of $D(\tilde{h})$. Since each eigenvalue is simple, thus every corresponding eigenfunction belongs to subspace $X_k$. One this subspace, we have
$$h[\varphi]-\tilde{h}[\varphi]=(\alpha-\tilde{\alpha})|\varphi(v)|^2.$$ \label{strict inequality}  Since each eigenfunction is non-zero at vertex $v$, and $\alpha-\tilde{\alpha} < 0.$ Therefore, for every eigenfunction, one has $h < \tilde{h}$, and the strict inequality follows from the min-max description of eigenvalues. Thus, for any $\epsilon >0$, no matter how small, if we let $\tilde{\alpha}=\alpha+\epsilon$. we have 
$$\lambda_k(\Gamma_{\alpha}) < \lambda_k(\Gamma_{\alpha+\epsilon}).$$
consolidating this inequality with 
$$\lambda_k(\Gamma_{\alpha+\epsilon}) \leq \lambda_k(\Gamma_{\tilde{\alpha}}),$$
 then for $\alpha+\epsilon < \tilde{\alpha}$, we obtain
 $$\lambda_k(\Gamma_\alpha) < \lambda_k(\Gamma_{\tilde{\alpha}}).$$}
\end{proof}

This shows that if one increases the total interaction strength (sum of all $\alpha_m $), then spectrum also increases. The above theorem also shows that whenever $\alpha < \tilde\alpha$, there is exactly one eigenvalue of a graph $\Gamma_{\tilde{\alpha}}$ between two consecutive eigenvalues of a graph $\Gamma_{\alpha}$. This useful observations leads to the following generalization of \textcolor{blue}{Theorem} \eqref{strength}.
\begin{theorem}
Let $\Gamma_{\tilde{\alpha}}$ be a finite compact graph obtained from $\Gamma_\alpha$ by replacing strengths from $\alpha_m$ to $\tilde{\alpha}_m$ at each vertex of $\Gamma$, if $|V|$ are the number of vertices in $\Gamma_\alpha$, then 
\begin{equation}
    \lambda_k(\Gamma_\alpha) \leq \lambda_k(\Gamma_{\tilde{\alpha}}) \leq \lambda_{k+|V|}(\Gamma_\alpha).
\end{equation}
\end{theorem}
Let $\Gamma^i$ for $i=1,2,3,4$ be metric graphs as discussed earlier, but the strength at the vertex $v_0$ in each graph is $\alpha_0=\infty$. We shall name them extended vertex conditions, and such conditions require that the function be zero at vertex $v_0$.
Let $h^i_\infty$ and $D(h^i_\infty)$ be corresponding quadratic forms and respective domains, and the inclusion \eqref{inclusions} and the inequalities \eqref{inequalities} still holds.
The extended conditions also imply the continuity of $\varphi$ in the form domain, moving to general conditions  with $\alpha < \infty$ simply means weakening the condition on $\varphi$ in form's domain, since now we only require continuity of $\varphi$ at $v_0$. Therefore, the domains $D(h^i_\infty)$ are smaller than the form domains of $h_{\alpha}$ and $h_{\tilde{\alpha}}$.
Let $\Gamma_{\alpha}$ and $\Gamma_{\tilde{\alpha}}$ be as in \textcolor{blue}{Theorem} \eqref{strength}, then for the functions $\varphi$ from $W^2_2(\Gamma\backslash V)$, which are also continuous on the whole graph, we have 
$$h^i_\alpha[\varphi]=h^i_\infty[\varphi]+\alpha |\varphi(v_0)|^2, \quad \text{and} \quad  h^i_{\tilde{\alpha}}[\varphi]=h^i_\infty[\varphi]+\tilde{\alpha} |\varphi(v_0)|^2,$$ 
Since the interaction strengths, $\alpha_m$ does not enter the form domain. Therefore, $D(h_\alpha)=D(h_{\tilde{\alpha}})$, and $D(h^1_\infty)=\{ \varphi \in D(h_{\alpha}): \varphi(v)=0 \}$, moreover, the quadratic forms $h^i_\infty$ and $h^i_{\tilde{\alpha}}$ agree on $D(h^i_{\infty}).$
Therefore, minimizing over smaller space makes the eigenvalue large. Thus, we obtain the following inequalities. $$\lambda_k(\Gamma_{\alpha}) \leq \lambda_k(\Gamma_{\tilde{\alpha}}) \leq \lambda_k(\Gamma^1_{\infty}) \leq \lambda_{k+1}(\Gamma_{\alpha}), \quad \quad \lambda_k(\Gamma_{\alpha}) \leq   \lambda_k(\Gamma_{\tilde{\alpha}}) \leq \lambda_k(\Gamma^2_{\infty}) \leq \lambda_{k+2}(\Gamma_{\alpha}),  $$
   $$\lambda_k(\Gamma_{\alpha}) \leq \lambda_k(\Gamma_{\tilde{\alpha}}) \leq \lambda_k(\Gamma^3_{\infty}) \leq \lambda_{k+d}(\Gamma_{\alpha}), \quad  \quad \lambda_k(\Gamma_{\alpha}) \leq   \lambda_k(\Gamma_{\tilde{\alpha}}) \leq \lambda_k(\Gamma^4_{\infty}) \leq \lambda_{k+d+1}(\Gamma_{\alpha}).  $$

\section{Joining and splitting of vertices}
In this section, we aim to examine the effect on the graph's spectrum, mainly due to the changes in the connectivity of the underlying graph. We discuss the behaviour of the spectrum of a graph when two of its vertices are joined (not necessarily of the same type) and integrate their coupling conditions befittingly. This procedure does not affect the total number of edges and therefore also preserves the total length of the graph; however, the graph's connectivity is increased.
\\ \\
For a finite and compact metric graph $\Gamma$ consider the operator $H=\frac{d^4}{dx^4}$ in $L^2(\Gamma)$ with any self-adjoint vertex conditions at all vertices except at two distinct vertices $v_1$  and $v_2$. Assume that the vertices $v_1$ and $v_2$ are equipped with conditions either \eqref{vc1}, \eqref{gvc1}, \eqref{gvc2} or \eqref{vc2} (not necessarily of the same type at $v_1$ and $v_2$) with interaction strengths $\alpha_1$ and $\alpha_2$, respectively. Let  $\tilde\Gamma$ is obtained by identifying $v_1$ and $v_2$ to obtain a new vertex $v_0$ with strength $\alpha_0=\alpha_1+\alpha_2$. Based on the types of conditions imposed at $v_1$ and $v_2$, these two vertices can be glued together in several different ways. Below, we list all the possibilities and divide them into three classes. The first two correspond to the rank one and two perturbations, and the third corresponds to the rank $d_2$ perturbation, where $d_2$ is the degree of vertex $v_2$.
\begin{enumerate}[I]
  \item \begin{enumerate}[1.]
  
\item Conditions \eqref{vc1} are imposed at $v_1$ and $v_2$. The glued vertex $v_0$ is also equipped with  conditions \eqref{vc1}.

\item Conditions \eqref{vc1} and \eqref{gvc1} are imposed at $v_1$ and $v_2$, respectively. The glued vertex $v_0$ is equipped with conditions \eqref{vc1}.

\item Conditions \eqref{vc1} and \eqref{gvc2} are imposed at $v_1$ and $v_2$, respectively. The glued vertex $v_0$ is equipped with conditions \eqref{vc1}.

\item Conditions \eqref{vc1} and \eqref{vc2} are imposed at $v_1$ and $v_2$, respectively. The glued vertex $v_0$ is equipped with conditions \eqref{vc1}.

\item Conditions \eqref{gvc1} are imposed at $v_1$ and $v_2$. The glued vertex $v_0$ is also equipped with conditions \eqref{gvc1}.

\item Conditions \eqref{gvc1} and \eqref{vc2} are imposed at $v_1$ and $v_2$, respectively. The glued vertex $v_0$ is equipped with conditions \eqref{gvc1}.

\item Conditions \eqref{vc2} are imposed at $v_1$ and $v_2$. The glued vertex $v_0$ is also equipped with conditions \eqref{vc2}.
\\
\end{enumerate}
\item \begin{enumerate}[1.]
  \item Conditions \eqref{gvc1} and \eqref{vc1} are imposed at $v_1$ and $v_2$, respectively. The glued vertex $v_0$  is equipped with conditions \eqref{gvc1}.
\item Conditions \eqref{gvc2} are imposed at $v_1$ and $v_2$. The glued vertex $v_0$ is also equipped with conditions \eqref{gvc2}.
\item Conditions \eqref{gvc2} and \eqref{vc2} are imposed at $v_1$ and $v_2$, respectively. The glued vertex $v_0$ is equipped with conditions \eqref{gvc2}.\\
\end{enumerate}

\item \begin{enumerate}[1.]
\item Conditions \eqref{gvc2} and \eqref{vc1} are imposed at $v_1$ and $v_2$, respectively. The glued vertex $v_0$ is equipped with conditions \eqref{gvc2}.
\item Conditions \eqref{vc2} and \eqref{vc1} are imposed at $v_1$ and $v_2$, respectively. The glued vertex $v_0$ is equipped with conditions \eqref{vc2}.
\item Conditions \eqref{vc2} and \eqref{gvc1} are imposed at $v_1$ and $v_2$, respectively. The glued vertex $v_0$ is equipped with conditions \eqref{vc2}.
\item Conditions \eqref{vc2} and \eqref{gvc2} are imposed at $v_1$ and $v_2$, respectively. The glued vertex $v_0$ is equipped with conditions \eqref{vc2}.
\\
\end{enumerate}

\end{enumerate}
\begin{figure}[H]
\centering 
\includegraphics[scale=1.0]{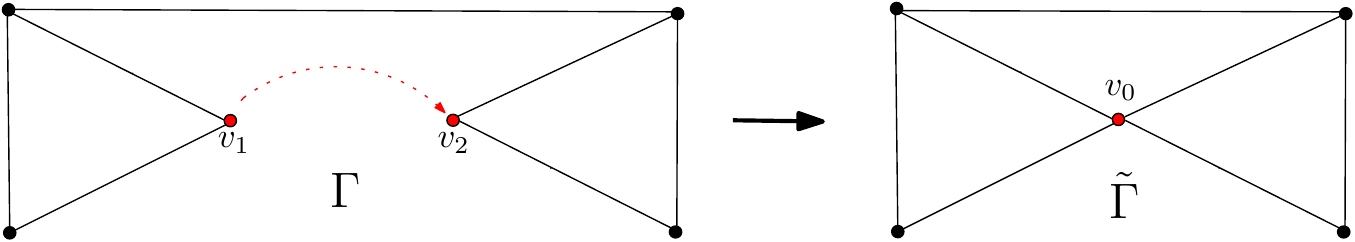}
\caption{The graph $\tilde\Gamma$ is formed by gluing the
vertices $v_1$ and $v_2$. Conversely, the graph  $\Gamma$ is one of the possible graphs obtainable from $\tilde\Gamma$ by splitting $v_0$, and producing the vertices $v_1, v_2$. }
\end{figure}
Let $\tilde H=\frac{d^4}{dx^4}$ be the self-adjoint operator in $L^2(\tilde\Gamma)$ subject to the same vertex conditions as  $H$ at all vertices $v\in V\backslash\{v_1,v_2\}$ and based on conditions at $v_1$ and $v_2,$  the glued vertex $v_0$ satisfy vertex conditions according to any one of the choices mentioned in the above three cases. The following theorem inspired by \cite{BKKM} shows that the gluing of vertices leads to a non-decrease of eigenvalues and specific interlacing properties held between the eigenvalues of $H$ and $\tilde H$.

\begin{theorem}\label{gluing}
\begin{enumerate}
\item If the vertices $v_1$ and $v_2$ of graph $\Gamma$ and the glued vertex $v_0$ of graph $\tilde \Gamma$ are equipped with vertex conditions according to the classifications I . Then the conclusions of \textcolor{blue}{Theorem} \eqref{berk1} hold for the eigenvalues of operators $H$ and $\tilde H$.
\\
\item If the vertices $v_1$ and $v_2$ of graph $\Gamma$ and the glued vertex $v_0$ of graph $\tilde \Gamma$ are equipped with vertex conditions according to the classifications II. Then,
$$
\lambda_k(\Gamma)\leq \lambda_k(\tilde{\Gamma})\leq \lambda_{k+2}(\Gamma)\leq \lambda_{k+2}(\tilde{\Gamma}), \quad k\geq 1.
$$ 

\item If the vertices $v_1$ and $v_2$ of graph $\Gamma$ and the glued vertex $v_0$ of graph $\tilde \Gamma$ are equipped with vertex conditions according to the classifications III. Then,
$$
\lambda_k(\Gamma)\leq \lambda_k(\tilde{\Gamma})\leq \lambda_{k+d_2+1}(\Gamma)\leq \lambda_{k+d_2+1}(\tilde{\Gamma}), \quad k\geq 1.$$
Where $d_2$ is the degree of vertex $v_2.$
\end{enumerate}
\end{theorem}
\begin{proof}
Let $D(h)$ and $D(\tilde{h})$ be domains of quadratic forms $h$ and $\tilde{h}$ corresponding to metric graphs $\Gamma$ and $\tilde{\Gamma}$, respectively. Since the gluing of two vertices impose conditions on domain $D(\tilde{h})$, thus the domain is restricted and $D(\tilde{h}) \subset_{2} D(h)$. The inequality $\lambda_k(\Gamma) \leq \lambda_k(\tilde{\Gamma})$ follows from the fact that minimising over  smaller set $D(\tilde{h})$ results in larger eigenvalues.  The inequality $\lambda_k(\tilde{\Gamma}) \leq \lambda_{k+2}(\Gamma)$ follows from rank two nature of perturbation.   
\end{proof}

\begin{corollary} \label{joiningg}
Let the vertices $v_1$ and $v_2$ of graph $\Gamma$ are glued together to obtain a new graph $\tilde\Gamma$. Denote the glued vertex by $v_0$. Suppose that for some $k\geq 1$ there exist eigenfunctions $\varphi_1,\cdots,\varphi_k$ corresponding to eigenvalues $\lambda_1(\Gamma),\cdots,\lambda_k(\Gamma)$, respectively, such that
\begin{equation}\label{levelpoints1}
\varphi_1(v_1)=\varphi_1(v_2), \cdots, \varphi_k(v_1)=\varphi_k(v_2).
\end{equation}
\begin{enumerate}
 \item If gluing is performed according to the classifications, $I$ then
\begin{equation}\label{levelpoints2}
\lambda_1(\tilde \Gamma)=\lambda_1(\Gamma),\cdots, \lambda_k(\tilde \Gamma)=\lambda_k(\Gamma).
\end{equation}
\item If gluing is performed according to the classification $II(1)$ and, in addition to \eqref{levelpoints1}, eigenfunctions also satisfy
\begin{equation}
    \sum\limits_{x_{j} \in v_{2}} \overline{\sigma^{v_{2}}_{x_{j}}} \partial \varphi(x_j)=0,
\end{equation}
then \eqref{levelpoints2} holds.
\item If gluing is performed according to the classification $II(2)$ and, in addition to \eqref{levelpoints1}, eigenfunctions also satisfy
\begin{equation} \label{levelpoints3}
\frac{\partial\varphi_1(x_i)}{\sigma_{x_i}^{v_1}}=\frac{\partial\varphi_1(x_j)}{\sigma_{x_j}^{v_2}},\cdots,\frac{\partial\varphi_k(x_i)}{\sigma_{x_i}^{v_1}}=\frac{\partial\varphi_k(x_j)}{\sigma_{x_j}^{v_2}}, \quad x_i\in v_1,x_j\in v_2,    
\end{equation}

then \eqref{levelpoints2} holds.

\item If gluing is performed according to the classification $II(3)$ and, in addition to \eqref{levelpoints1}, eigenfunctions also satisfy
\begin{equation}
     \partial \varphi(x_1)=0, \quad x_1 \in v_1
\end{equation}
then \eqref{levelpoints2} holds.

\item If gluing is performed according to the classification $III(2,3, 4)$ and, in addition to \eqref{levelpoints1}, eigenfunctions also satisfy
\begin{equation}
     \partial \varphi(x_j)=0, \quad x_j \in v_2
\end{equation}
then \eqref{levelpoints2} holds.
\end{enumerate}
\end{corollary}
\begin{proof}
\begin{enumerate}
    \item We will prove it for $(1)$ only; the same argument holds in other cases. 
    Since, the domain of quadratic form for $\tilde{\Gamma}$ only involves continuity of functions at vertices, so $\varphi_k(x)$ is also in domain of quadratic form for $\tilde\Gamma$, and if in addition, $\varphi_k(x)$ satisfies (\ref{levelpoints1}), then, $\varphi_k(x)$ is also a minimizer for $\tilde{\Gamma}$, and corresponding eigenvalues coincide.
\end{enumerate}

\end{proof}
{
\begin{remark}\noindent
In above theorem, we have only considered case in which glued vertex $v_0$ is equipped with condition that is imposed on either $v_1$ or $v_2$. However, one can consider the case when glued vertex is equipped with condition, which is neither imposed on $v_1$ nor $v_2$. To this end, we describe the relation between eigenvalues of graphs $\Gamma$ and $\tilde{\Gamma}$ when $v_1$ and $v_2$ are equipped with \eqref{vc1} and $v_0$ is equipped with  \eqref{gvc1}. Let $\Gamma'$ be another graph obtained from $\Gamma$ by gluing $v_1$ and $v_2$, in accordance to \textcolor{blue}{Theorem} \eqref{gluing} when glued vertex $v_0^*$ is equipped with condition  \eqref{vc1}. Then $$\lambda_k(\Gamma) \leq \lambda_k(\Gamma')\leq \lambda_{k+1}(\Gamma).$$ 
Now, create a graph $\tilde{ \Gamma}$ from $\Gamma'$ by replacing the condition at $v_0^*$  from \eqref{vc1} to \eqref{gvc1} using \textcolor{blue}{Theorem} \eqref{changeofcondition}. Thus, we get
$$\lambda_k(\Gamma') \leq \lambda_k (\tilde{\Gamma}) \leq \lambda_{k+1}(\Gamma' ).$$
using the above two inequalities, we obtain
$$\lambda_k(\Gamma) \leq \lambda_k(\tilde{ \Gamma})\leq \lambda_{k+2}(\Gamma).$$
Similarly, one can list all possible interlacing inequalities for other cases also.
\end{remark}}
One can observe that if the form $\tilde h$ is a positive rank-m perturbation of the form $h$, that it, if $D(\tilde h)$ is co-dimension $m$ subspace of $D(h)$, then we obtain following inequality by applying the same arguments.
$$\lambda_k(\Gamma) \leq \lambda_k(\tilde{\Gamma}) \leq \lambda_{k+m}(\Gamma).$$
This observation gives rise to the following generalization of Theorem \eqref{gluing}. 
\begin{theorem} \label{joining}
\begin{enumerate}
    \item If the vertices $v_0,v_1, v_2,\cdots,v_m$ of graph $\Gamma$ and the glued vertex $v^*_0$ of graph $\tilde\Gamma$ are equipped with vertex conditions according to the classification I. Then
\begin{equation}
    \lambda_k(\Gamma) \leq \lambda_k(\tilde{\Gamma}) \leq \lambda_{k+m}(\Gamma).
\end{equation}
\item If the vertices $v_0,v_1, v_2,\cdots,v_m$ of graph $\Gamma$ and the glued vertex $v^*_0$ of graph $\tilde\Gamma$ are equipped with vertex conditions according to the classification II. Then
\begin{equation}
    \lambda_k(\Gamma) \leq \lambda_k(\tilde{\Gamma}) \leq \lambda_{k+2m}(\Gamma).
\end{equation} 
\end{enumerate}
\end{theorem}
\begin{proof}
The gluing of $m+1$ vertices at one vertex is the same as pairwise gluing of $m$ pair of vertices. At each step, we identify two vertices and add parameters $\alpha_m$. By the repeated application of \textcolor{blue}{Theorem} \eqref{gluing}, we can achieve the desired inequality.
\end{proof}
\begin{definition}
Let $\Gamma$ be finite, compact connected metric graph. A \textbf{flower graph} $\Gamma_f$ is constructed from $\Gamma$ by attaching together both endpoints of all edges at a single vertex and  assign total strength $\alpha=\sum\limits_{m=1}^{|V|}\alpha_m$ at that vertex. 

\end{definition}
\begin{figure}[H]
\centering 
\includegraphics[scale=1.0]{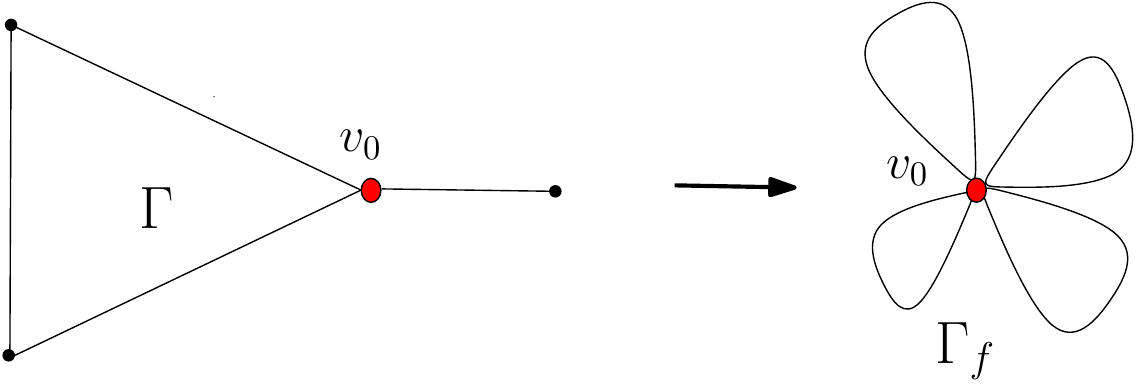}
\caption{A flower graph $\Gamma_f$ is obtained by identifying all vertices of the graph $\Gamma$ at $v_0$.}
\end{figure}

\begin{corollary} \label{flower}
 Let $\Gamma$ be a finite compact connected metric graph with vertex set $V$ and let $\Gamma_{f}$ be the corresponding flower graph.
\begin{enumerate}
    \item If each vertex of $\Gamma$ is equipped with vertex conditions either \eqref{vc1}, \eqref{gvc1}, or \eqref{vc2}.  Then
    \begin{equation*}
        \lambda_k(\Gamma) \leq \lambda_k(\Gamma_{f})\leq \lambda_{k+|V|-1}(\Gamma).
    \end{equation*} 
    \item If each vertex of $\Gamma$ is equipped with vertex conditions \eqref{gvc2}.  Then
    \begin{equation*}
        \lambda_k(\Gamma) \leq \lambda_k(\Gamma_{f})\leq \lambda_{k+2|V|-2}(\Gamma).
    \end{equation*} 
\end{enumerate} 

\end{corollary}
\begin{proof}
 Since, in each step, we glue two points of an edge, so at each step, we restrict the domain of the quadratic form, and this operation lowers all eigenvalues.
\end{proof}
\begin{figure}[H]
\centering 
\includegraphics[scale=1.2]{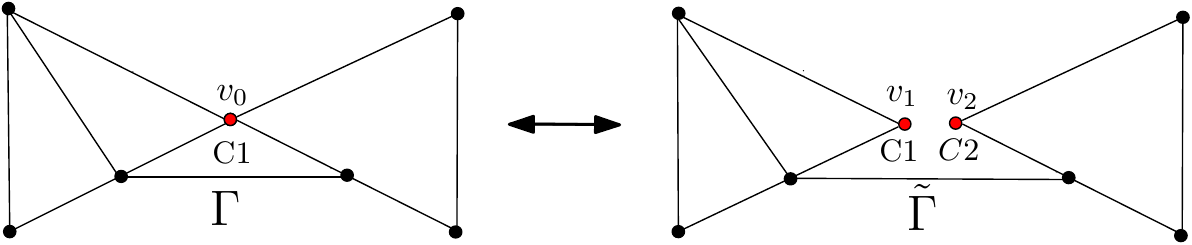}
\caption{Spilitting the vertex $v_0$ into two vertices $v_1$ and $v_2$, we obtain   $\tilde{\Gamma}$. Gluing the vertices $v_1$ and $v_2$ to form single vertex $v_0$, we obtain  $\Gamma$.}
\label{cut}
\end{figure}
The converse of gluing operation is splitting a vertex into two or more vertices. Let $\Gamma$ and $\tilde{\Gamma}$ be two finite and compact metric graphs, as shown in figure \ref{cut}. Let $v_0, v_1$ be vertices equipped with vertex conditions \eqref{vc1}, and $v_2$ be equipped with conditions \eqref{gvc1}. Since the expressions for quadratic form for both graphs are the same. Therefore, we need to compare their respective domains. A graph with a larger domain will have smaller eigenvalues. We can view figure in two ways. First, the graph $\Gamma$ is obtained from $\tilde{\Gamma}$ by gluing vertices $v_1$ and $v_2$. Second, the graph $\tilde{\Gamma}$ is obtained from $\Gamma$ by splitting vertex $v_0$, producing two vertices $v_1$ and $v_2$. In the former case, when gluing vertices together, the domain is restricted, and eigenvalues become large. Since splitting is converse to gluing, in the latter case, splitting should make the domain of quadratic form for $\tilde{\Gamma}$ larger and hence smaller eigenvalues. However, any continuous function on $\Gamma$ is also continuous on $\tilde{\Gamma}$, but in addition to continuity, we require this function to satisfy second equation in \eqref{gvc1} at $v_2$.
\begin{figure}[H]
\centering 
\includegraphics[scale=1.2]{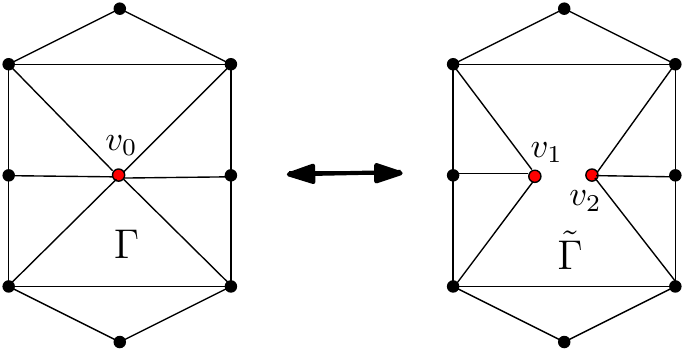}
\caption{Spilitting the vertex $v_0$ into two vertices $v_1$ and $v_2$, we obtain the graph  $\Gamma'$. Gluing the vertices $v_1$ and $v_2$ to form single vertex $v_0$, we obtain the graph $\Gamma$.}
\label{split}
\end{figure}

We compare  Laplacian to Bi-Laplacian to have a more precise idea of why splitting a vertex for Bi-Laplacian with specific conditions does not behave as a converse of gluing of a vertex. Let $\Gamma$ be  finite and compact quantum graph obtained from  $\tilde{\Gamma}$ by gluing $v_1$ and $v_2$, and conversely, let $\tilde{\Gamma}$ be retrieved from $\Gamma$ by splitting vertex $v_0$ (as shown in the \textcolor{blue}{Figure} \ref{split}). Assume that the vertices $v_0$, $v_1$ and $v_2$ are equipped with $\delta$ conditions, and the Laplace operator acts on each edge of $\Gamma$ and $\tilde{\Gamma}$.  The quadratic form's domain contains continuous functions on the whole graph. Therefore, any function from $D(\tilde{h})$ with condition $\varphi(v_1)=\varphi(v_2)$ is suitable for trial function on $\Gamma$. Conversely, every eigenfunction on $\Gamma$ can be lifted on $\tilde{\Gamma}$, and either this function will be an eigenfunction for $\tilde{\Gamma}$, or there will be some other eigenfunction. In both cases, we have following inequality and corresponding domains. 
\begin{align*}
\lambda_k(\tilde{\Gamma}) &\leq \lambda_k(\Gamma), \\ 
    D(h)&=\{\varphi \in D(\tilde{h}) : \varphi(v_1)=\varphi(v_2)\}.     
\end{align*}
Now assume that vertices $v_0, v_1$ and $ v_2$ are equipped with conditions \eqref{gvc1}, and Bi-Laplace operator acts on each edge. The  form domain consists of continuous functions on the whole graph and satisfies condition $\sum_{x_j\in v_m}\overline{\sigma^{v_m}_{x_j}}\partial\varphi(x_j)=0$. Therefore, any function on $\tilde{\Gamma}$ satisfying $\varphi(v_1)=\varphi(v_2)$ can be used as a test function on $\Gamma$, and the domain $D(h)$ is a co-dimension one subspace of $D(\tilde{h})$.  Now, consider the case when splitting is performed. Every eigenfunction on $\Gamma$ is also continuous on $\tilde{\Gamma}$. But to be a test function on $\tilde{\Gamma}$, we require this function to satisfy second equation in \eqref{gvc1} at both vertices $v_1$ and $v_2$. However, the eigenfunction satisfying equation $\sum_{x_j\in v_0}\overline{\sigma^{v_0}_{x_j}}\partial\varphi(x_j)=0$ at $v_0$  is unfitted to satisfy $\sum_{x_j\in v_1}\overline{\sigma^{v_1}_{x_j}}\partial\varphi(x_j)=0$ at $v_1$  and $\sum_{x_j\in v_2}\overline{\sigma^{v_2}_{x_j}}\partial\varphi(x_j)=0$ at $v_2$ . Therefore, we can not lift any function from $\Gamma$ to $\tilde{\Gamma}$. The converse of \textcolor{blue}{Theorem} \eqref{gluing} is in general not true. Instead we have the following result:

\begin{theorem}\label{splittingtheorem}
 Consider a vertex $v_0$ of a compact and finite metric graph $\Gamma$ with delta interaction of type either \eqref{vc1}, \eqref{gvc1}, \eqref{gvc2} or \eqref{vc2} with strength $\alpha_0$ and let another finite and compact metric graph $\tilde{\Gamma}$ is obtained by splitting the vertex $v_0$ into two vertices $v_1$ and $v_2$. The two new vertices can be equipped with either of the vertex conditions \eqref{vc1}, \eqref{gvc1}, \eqref{gvc2} or \eqref{vc2} (not necessarily of the same type as $v_0$) with corresponding strengths $\alpha_1$ and $\alpha_2$ such that $\alpha_0=\alpha_1+\alpha_2$.
Let us denote by $H$ and $\tilde{H}$ the self-adjoint differential operators $\frac{d^4}{dx^4}$  on the graph $\Gamma$ and $\tilde{\Gamma}$ with the same types and strengths of the vertex conditions at all preserved vertices. Let the vertex $v_0$ is splitted into $v_1$ and $v_2$ in any one of the following ways
\begin{enumerate}
\item Conditions \eqref{vc1} are imposed at $v_0$ and its descendants $v_1$ and $v_2$  are also equipped with conditions \eqref{vc1}.
\item Conditions \eqref{gvc2} are imposed at $v_0$ and its descendant vertices $v_1$ and $v_2$ are also equipped with conditions \eqref{gvc2}.
\item Conditions \eqref{vc2} are imposed at $v_0$ and its descendant vertices $v_1$ and $v_2$ are also equipped with conditions \eqref{vc2}.
\item Conditions \eqref{vc2} are imposed at $v_0$ and its descendant vertices $v_1$ and $v_2$ are equipped with conditions \eqref{vc1} and \eqref{vc2}, respectively.
\item Conditions \eqref{gvc2} are imposed at $v_0$ and its descendant vertices $v_1$ and $v_2$ are equipped with conditions \eqref{vc1} and \eqref{gvc2}, respectively.
\item Conditions \eqref{vc2}  are imposed at $v_0$ and its descendant vertices $v_1$ and $v_2$ are equipped with conditions \eqref{gvc1} and \eqref{vc2}, respectively.
\item Conditions \eqref{vc2} are imposed at $v_0$ and its descendant vertices $v_1$ and $v_2$ are equipped with conditions \eqref{gvc2} and \eqref{vc2}, respectively.

\end{enumerate}
Then the eigenvalues of $\Gamma$ and $\tilde{\Gamma}$ satisfy
\begin{equation}
\lambda_k(\Gamma)\geq\lambda_k(\tilde{\Gamma}).
\end{equation}
\end{theorem}

\begin{proof}
  We consider quadratic forms corresponding to $\Gamma$ and $\tilde{\Gamma}$ and perceive that they are defined precisely by the same expression, and the splitting of a vertex makes the domain of the quadratic form of the graph $\tilde{\Gamma}$ larger.  
Let's prove the case when $v_0$ and its descendant vertices $v_1, v_2$ are equipped with conditions \eqref{vc1}. Let $D(h)$ and $D(\tilde{h})$ denotes the domain of the quadratic forms of $\Gamma$ and $\tilde{\Gamma}$, respectively.
$$D(h)=\{\varphi \in D(\tilde{h}) : \varphi(v_1)=\varphi(v_2) \}$$
 Therefore, $D(h)$ is a  subspace of $D(\tilde{h})$. The same argument holds in all other cases.
\end{proof}
\section{Increasing volume of a graph}
In this section, we focus our study on understanding the precise nature of the relationships between eigenvalues of a graph and  the operation that increase the total volume of a graph, either by attaching a new subgraph to it or by scaling a part of a graph. These basic changes help to transform a graph into another graph such that the effect on one or several eigenvalues is predictable. The impact of surgery operation (increasing length of an edge, attaching pendant edge, insertion of a graph) on eigenvalues is different for a different set of vertex conditions. Moreover, the results and idea of proof are borrowed from \cite{BKKM}.

\subsection{Attaching a pendant graph}
Consider a vertex $v$ of a  graph $\Gamma$ and a vertex $w$ of a graph $\hat\Gamma$, both are equipped with any of the vertex conditions \eqref{vc1}, \eqref{gvc1}, \eqref{gvc2} or \eqref{vc2} with interaction strengths $\alpha_v$ and $\alpha_w$, respectively. All other vertices of $\Gamma$ and $\hat\Gamma$ are equipped with arbitrary  self-adjoint vertex conditions. Let graph $\tilde\Gamma$ is obtained by gluing together $v$ and $w$, according to either of the classifications I, II or III of section 4. The new vertex, say $v_0$, is equipped with any of the vertex conditions \eqref{vc1}, \eqref{gvc1}, \eqref{gvc2} or \eqref{vc2} with interaction strength  $\alpha_{v_0}=\alpha_v+\alpha_w$. We now say that the graph $\tilde\Gamma$ is obtained by attaching a  pendant graph $\hat\Gamma$ to the graph $\Gamma$. 
\begin{figure}[H]
\centering 
\includegraphics[scale=0.88]{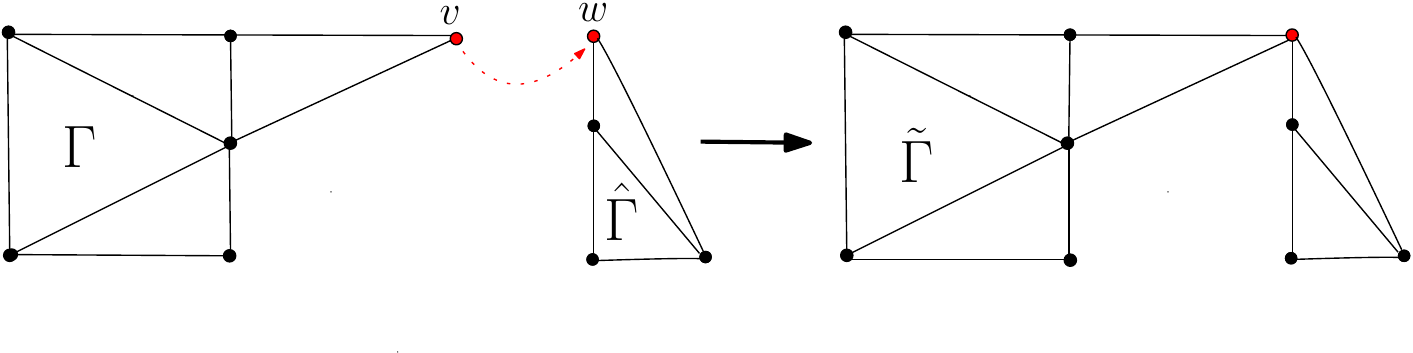}
\caption{Identifying $v, w$ together, we can glue the graph $\hat{\Gamma}$ to $\Gamma$, thus obtaining the graph $\tilde\Gamma$.}
\end{figure}

\begin{theorem}\label{pendant}
Let $\tilde\Gamma$ is formed from $\Gamma$ by attaching a vertex $w$ of a  pendant graph $\hat\Gamma$ at a vertex $v$ of $\Gamma$ and consider self-adjoint operators $H$, $\hat H$ and $\tilde H$ all given by the differential expression $\frac{d^4}{dx^4}$ defined in $L^2(\Gamma)$, $L^2(\hat \Gamma)$ and $L^2(\tilde\Gamma)$, respectively. Assume that for some $r$ and $k_{0}$
$$
\lambda_r(\hat \Gamma)\leq \lambda_{k_{0}}(\Gamma).
$$
\begin{enumerate} 
\item If the vertices $v$ and $w$ are glued according to classifications I then
$$
\lambda_{k+r-1}(\tilde \Gamma)\leq\lambda_k(\Gamma), \quad k\geq k_0.
$$
\item If the vertices $v$ and $w$ are glued according to classifications II, then
$$
\lambda_{k+r-2}(\tilde \Gamma)\leq\lambda_k(\Gamma), \quad k\geq k_0.
$$
\item If the vertices $v$ and $w$ are glued according to classifications III, then
$$
\lambda_{k+r-d-1}(\tilde \Gamma)\leq\lambda_k(\Gamma), \quad k\geq k_0.
$$
Where $d$ is the degree of vertex $v$. Moreover, if $\lambda_{k-1}( \Gamma) < \lambda_k(\Gamma)$ and $\lambda_r(\hat \Gamma) <  \lambda_{k_{0}}(\Gamma)$ and the eigenvalue $\lambda_k(\Gamma)$ has an eigenfunction which is non-zero at $v_0$, then the above three inequalities are strict.
\end{enumerate}

\end{theorem}
\begin{proof}
\begin{enumerate}
    \item Assume that for some $r$ and $k_0$, we have, $\lambda_r(\hat \Gamma)\leq \lambda_{k_{0}}(\Gamma)$. Let $\check{\Gamma}$ be a graph considered as a union of two disconnected components, $\Gamma$ and $\hat{\Gamma}$. Since the spectrum of the union of two graphs is equal to the union of their spectrum. Therefore, $\lambda_{k_{0}}(\Gamma)= \lambda_m(\Gamma \cup \hat\Gamma) $ for some $m \geq k_0+r$. Now, attach the vertex $w$ of $\hat\Gamma$ with the vertex $v$ of $\Gamma$ to obtain a new graph $\tilde\Gamma$. By \textcolor{blue}{Theorem} \eqref{gluing}, we obtain the following interlacing inequalities.
$$\lambda_{m-1}(\Gamma \cup \hat\Gamma) \leq \lambda_{m-1}(\tilde\Gamma) \leq \lambda_{m}(\Gamma \cup \hat\Gamma) \leq \lambda_m(\tilde \Gamma),$$
combining the estimate $m \geq k_{0}+r$ with $\lambda_{m-1}(\tilde\Gamma) \leq \lambda_{m}(\Gamma \cup \hat\Gamma)$, we get
$$\lambda_{k_{0}+r-1}(\tilde{\Gamma}) \leq \lambda_{k_0}(\Gamma).$$
Since $\lambda_r(\hat \Gamma)\leq \lambda_{k_{0}+i}(\Gamma)$, for $i=0,1,2,\cdots$, therefore, repeating the same process we obtain
$$\lambda_{k_{0}+i+r-1}(\tilde{\Gamma}) \leq \lambda_{k_0+i}(\Gamma). $$ Thus, 
$$\lambda_{k+r-1}(\tilde{\Gamma}) \leq \lambda_k(\Gamma), \quad k\geq k_0.$$

Let $\lambda_{m-1}(\Gamma \cup \hat{\Gamma}) < \lambda_m(\Gamma \cup \hat{\Gamma})=\lambda_k(\Gamma)$ be the first occurence of $\lambda_k(\Gamma)$ in the spectrum of $\Gamma \cup \hat{\Gamma}$. Since 
$\lambda_{k-1}( \Gamma) < \lambda_k(\Gamma)$ and $\lambda_r(\hat \Gamma) <  \lambda_{k_{0}}(\Gamma)$ we still have $m \geq k+r$. Suppose that $\lambda_{m-1}(\Gamma \cup \hat\Gamma) \leq \lambda_{m-1}(\tilde\Gamma)= \lambda_{m}(\Gamma \cup \hat\Gamma),$ this shows that the eigenfunction corresponding to the eigenvalue $\lambda_m(\Gamma \cup \hat{\Gamma})$ is also an eigenfunction of $\tilde{\Gamma}$. But the eigenspace of $\lambda_m(\Gamma \cup \hat{\Gamma})$ contains a function which vanishes identically on $\hat{\Gamma}$ and is non-zero at $v_0$. But this eigenfunction can not be contained in the eigenspace of $\tilde{\Gamma}$ because it does not belong to its form domain. Thus, $\lambda_{m-1}(\tilde{\Gamma}) < \lambda_m(\Gamma \cup \hat{\Gamma})=\lambda_k(\Gamma)$ and $\lambda_{k+r-1}(\tilde{\Gamma}) \leq \lambda_{m-1}(\tilde{\Gamma}) < \lambda_k(\Gamma)$.

\item Assume that for some $r$ and $k_0$, we have, $\lambda_r(\hat \Gamma)\leq \lambda_{k_0}(\Gamma)$. Let $\check{\Gamma}$ be a graph considered as a union of two disconnected components, $\Gamma$ and $\hat{\Gamma}$. Since the spectrum of the union of two graphs is equal to the union of their spectrum. Therefore, $\lambda_{k_0}(\Gamma)= \lambda_m(\Gamma \cup \hat\Gamma) $ for some $m \geq k_0+r$. Now, attach the vertex $w$ of $\hat\Gamma$ with the vertex $v$ of $\Gamma$ to obtain a new graph $\tilde\Gamma$. By \textcolor{blue}{Theorem}\eqref{gluing}, we obtain the following interlacing inequalities.
$$\lambda_{m-2}(\Gamma \cup \hat\Gamma) \leq \lambda_{m-2}(\tilde\Gamma) \leq \lambda_{m}(\Gamma \cup \hat\Gamma) \leq \lambda_{m}(\tilde \Gamma),$$
combining the estimate $m \geq k+r$ with $k\geq k_0$, and  $\lambda_{m-2}(\tilde\Gamma) \leq \lambda_{m}(\Gamma \cup \hat\Gamma)$, we get
$$\lambda_{k+r-2}(\tilde{\Gamma}) \leq \lambda_k(\Gamma).$$
\item This can be proved with the same argument.
\end{enumerate}
\end{proof}

\subsection{Inserting a graph at a vertex}
Consider a vertex $v_0$, equipped with any of the vertex conditions \eqref{vc1}, \eqref{gvc1}, \eqref{gvc2} or \eqref{vc2} with interaction strengths $\alpha_{v_0}$, of a metric graph $\Gamma$ having set of incident edges $\{e_1, e_2,\cdots,e_n\}$. Let another metric graph $\hat\Gamma$ with a subset of vertices $\{w_1,\cdots,w_k\}\subseteq V(\hat\Gamma)$, $k\leq n$. Let $\tilde\Gamma$ is a new graph formed from $\Gamma$ by first removing the vertex $v_0$ of $\Gamma$ and then attaching the incident edges $e_i$, $i=1,\cdots ,n$ to the vertices $w_j$, $j=1,\cdots,k$ of $\hat\Gamma$ in such a way that the newly formed vertices, say $w_j'$, $j=1,\cdots,m$, are equipped with either of the vertex conditions \eqref{vc1}, \eqref{gvc1}, \eqref{gvc2} or \eqref{vc2} with total interaction strength equal to $\alpha_{v_0}$. We call $w_j'$ the \textit{post insertion descendents} (PIDs) of $v_0$ and say that $\tilde\Gamma$ is formed by inserting $\hat\Gamma$ into $\Gamma$ at $v_0$ .
\begin{figure}[H]
\centering 
\includegraphics[scale=1.0]{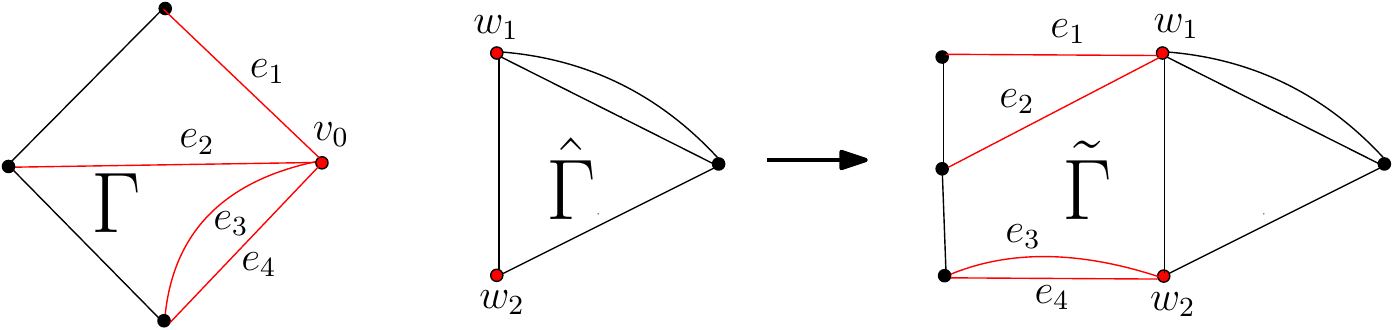}
\caption{Inserting $\hat\Gamma$ into $\Gamma$ at $v_0$, we obtain the graph $\tilde\Gamma.$}
\end{figure}

\begin{theorem}  \label{insertprop}
Let conditions \eqref{vc2} were imposed at all vertices of $\hat\Gamma$ with interaction strengths equal to zero prior to insertion. Then, for all $k$ such that $\lambda_k(\Gamma)\geq 0$ and 

\begin{center}
\begin{tabular}{ |c|c| c|c|} 
 \hline
 $v_0$ & $w'_1$ & $w'_2$ & Classification \\ 
 \hline
 $C1$ & $C1$ & $C1$ & $I$  \\
 \hline
 $C4$ & $C4$ & $C1, \ C2, \  C3,\  C4 $ & $I$  \\
 \hline
$C3$ & $C3$ & $C1, \ C3 $ & $II$ \\
\hline
$C1$ & $C4$ & $C1, \ C2, \  C3,\  C4 $ & $III$  \\
 \hline
 $C2$ & $C4$ & $C1, \ C2, \  C3,\  C4 $ & $III$  \\
 \hline
 $C3$ & $C4$ & $C1, \ C2, \  C3,\  C4 $ & $III$  \\
 \hline 
\end{tabular}
\end{center}
\begin{figure} [H]
    \centering
    \includegraphics[scale=0.7]{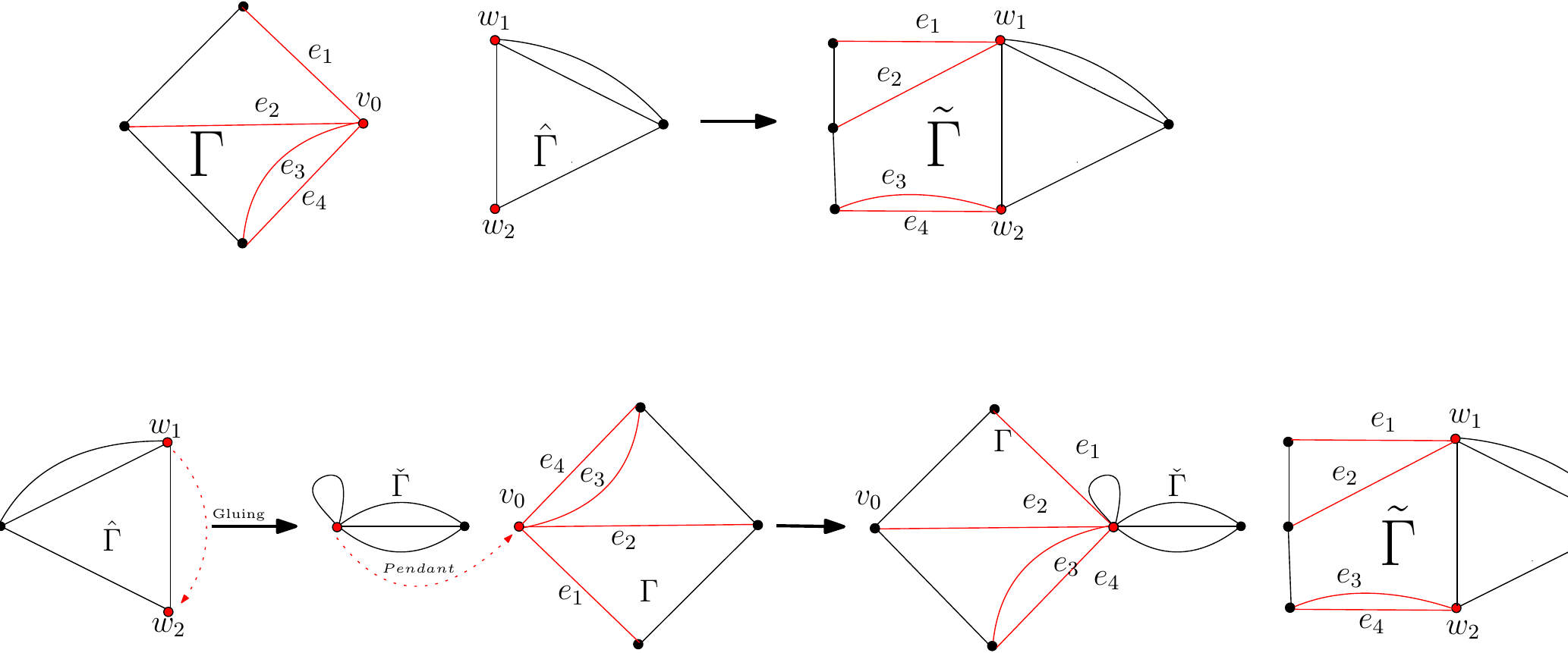}
    \caption{Steps of the proof of Theorem (6.2).}
    \label{fig:my_label}
\end{figure}
\begin{enumerate}
\item  for the cases described in classification $I$ we have
$$
\lambda_k(\tilde \Gamma)\leq \lambda_k(\Gamma).
$$
\item for the cases described in classification $II$ we have
$$
\lambda_k(\tilde \Gamma)\leq \lambda_{k+1}(\Gamma).$$
\item for the cases described in classification $III$ we have
$$
\lambda_k(\tilde \Gamma)\leq \lambda_{k+d}(\Gamma).$$
\end{enumerate}
Moreover, if the eigenfunction in non zero at $v_0$ and $\lambda_k(\Gamma) > \max (0, \lambda_{k-1}(\Gamma))$, then above three inequalities are strict.
\end{theorem}
\begin{proof}
\begin{enumerate}
\item   Let $\check \Gamma$ is obtained by gluing the vertices $w_1$ and $w_2$ of $\hat \Gamma$ in such a way that the glued vertex, $w^*$, is also equipped with conditions \eqref{vc2}. As all the vertices of $\check \Gamma$ are equipped with vertex conditions \eqref{vc2} with zero delta interactions , therefore $\lambda_1(\check \Gamma)=0$. Assumption $\lambda_k(\Gamma)\geq 0$ implies $\lambda_k(\Gamma)\geq\lambda_{r=1}(\check \Gamma)$.  We obtain a new graph $\Gamma'$ by gluing vertex $v_0$ of $\Gamma$, which is equipped with conditions \eqref{vc1} with delta interaction strength $\alpha$, and $w^*$ of $\check \Gamma$ such that the new vertex  $v_0^*$ is equipped with \eqref{vc1} with delta interaction strength $\alpha$. \textcolor{blue}{Theorem} \ref{pendant}(1) with $r=1$ gives $\lambda_k(\Gamma') \leq \lambda_k(\Gamma)$. Graph $\tilde \Gamma$ is created by splitting the vertex $v_0^*$ of $\Gamma'$ and restoring the vertices $w'_1$ and $w'_2$ which, now, are equipped with conditions \eqref{vc1} and total delta interaction strength equal to $\alpha$. \textcolor{blue}{Theorem} \eqref{splittingtheorem} implies $\lambda_k(\tilde \Gamma)\leq \lambda_k(\Gamma')$.
The case of strict inequality is obtained by using the strict version of \textcolor{blue}{ Theorem} \eqref{pendant}. 

        \item Let $\check \Gamma$ be a graph obtained by gluing the vertices $w_1$ and $w_2$ of $\hat \Gamma$ in such a way that the glued vertex, $w^*$, is also equipped with conditions \eqref{vc2}. As all the vertices of $\check \Gamma$ are equipped with vertex conditions \eqref{vc2} with zero delta interactions , therefore $\lambda_1(\check \Gamma)=0$. Assumption $\lambda_k(\Gamma)\geq 0$ implies $\lambda_k(\Gamma)\geq\lambda_{r=1}(\check \Gamma)$.  We obtain a new graph $\Gamma'$ by gluing vertex $v_0$, which is equipped with conditions \eqref{gvc2} with delta interaction strength $\alpha$, of $\Gamma$ and $w^*$ of $\check \Gamma$ in such a way that new vertex  $v_0^*$ is equipped with \eqref{gvc2} and delta interaction strength $\alpha$. \textcolor{blue}{Theorem} \ref{pendant}(2) with $r=1$ gives $\lambda_{k-1}(\Gamma')\leq \lambda_k(\Gamma)$. Graph $\tilde \Gamma$ is created by splitting the vertex $v_0^*$ of $\Gamma'$ and restoring the vertices $w'_1$ and $w'_2$ which, now, are equipped with conditions \eqref{gvc2} and total delta interaction strength equal to $\alpha$. \textcolor{blue}{Theorem} \eqref{splittingtheorem} implies $\lambda_{k-1}(\tilde \Gamma)\leq \lambda_{k-1}(\Gamma')$.
        \item This can be proved similarly.
        \end{enumerate}
\end{proof}                                                          
\begin{corollary}
 Suppose there exist two vertices $v$ and $w$ of $\Gamma$ and first $n$ eigenfunctions $\varphi_1,\varphi_2,\cdots,\varphi_n$ such that  
\begin{equation}\label{conti}
\varphi_1(v)=\varphi_1(w), \cdots, \varphi_n(v)=\varphi_n(w).
\end{equation}
Let both vertices are equipped with either vertex conditions \eqref{vc1} or \eqref{vc2} (of the same type on both $v$ and $w$).  If $\lambda_k(\Gamma)\geq 0$ for all $k=1,\cdots,n$, then the graph $\tilde\Gamma$ obtained by inserting an edge of arbitrary length between $v$ and $w$, with preserved type and interaction strength of vertex conditions at both vertices, satisfies
$$
\lambda_k(\tilde \Gamma)\leq \lambda_k(\Gamma), \quad k=1,\cdots,n.
$$
If $v$ and $w$ are both equipped with vertex conditions \eqref{gvc2} then we require that the first $n$ eigenfunctions, in addition to \eqref{conti}, also satisfy
$$
\frac{\partial\varphi_1(x_i)}{\sigma_{x_i}^{v}}=\frac{\partial\varphi_1(x_j)}{\sigma_{x_j}^{w}},\cdots,\frac{\partial\varphi_n(x_i)}{\sigma_{x_i}^{v}}=\frac{\partial\varphi_n(x_j)}{\sigma_{x_j}^{w}}, \quad x_i\in v,x_j\in w.
$$
In that case, if $\lambda_k(\Gamma)\geq 0$ for all $k=1,\cdots,n$, then the graph $\tilde\Gamma$ obtained by inserting an edge of arbitrary length between $v$ and $w$, with preserved type and interaction strength of vertex conditions at both vertices, satisfies
$$
\lambda_k(\tilde \Gamma)\leq \lambda_{k+1}(\Gamma), \quad k=1,\cdots,n.
$$
\end{corollary}
\begin{proof}
 Let both vertices $v$ and $w$ are equipped with vertex conditions \eqref{vc1}. Let $\Gamma'$ be a graph obtained from $\Gamma$ by gluing the vertices $v$ and $w$ in accordance with definition of joining of vertex; let $v^*$ be glued vertex. Then by \textcolor{blue}{Corollary} \eqref{joiningg}(1), 
    $$\lambda_k({\Gamma'}) = \lambda_k(\Gamma), \quad k=1,2,\cdots n.$$ 
    Now, insert a one edge graph $\hat\Gamma$ satisfying the conditions stipulated in \textcolor{blue}{Theorem} \eqref{insertprop}(1)  at vertex $v^*$; this graph coincides with $\tilde\Gamma$, and we get, $\lambda_k(\tilde \Gamma) \leq \lambda_k(\Gamma')=\lambda_k(\Gamma), \quad k=1,2,\cdots n.$ \\
    When $v$ and $w$ are both equipped with vertex conditions \eqref{gvc2}, the proof is identical to the above proof; however, we use corollary \eqref{joining}(2) and  \textcolor{blue}{Theorem}\eqref{insertprop}(2).
\end{proof}
{
The reason for imposing condition \eqref{vc2} on vertices $w_1$ and $w_2$ are twofold. First, during the proof, at one time, we split the vertex $v_0^*$, and in general, we know that splitting is not behaving exactly as a converse of gluing. However, the two are compatible when conditions \eqref{vc2} are specified. Second, at every step of the surgery of a graph, the proof does not involve the choices for more conditions. To understand the second reason, we assume that $v_0$ is equipped with condition \eqref{vc1}, and vertices $w_1$ and $w_2$ are endowed with conditions \eqref{vc1} and \eqref{gvc1}, respectively. We have two choices for $w^*$ either it is equipped with condition \eqref{vc1} or \eqref{gvc1}. If the vertex $w^*$ is equipped with \eqref{vc1}, then the vertex $v^*_0$ and the vertices $w_1$ and $w_2$ after the insertions are also equipped with condition \eqref{vc1}, and we have the following inequality.
$$\lambda_{k+r-1}(\tilde \Gamma)\leq \lambda_{k}(\Gamma).$$
When $w^*$ is equipped with condition \eqref{gvc1}, now, we have two choices for $v^*_0$. It is either equipped with condition \eqref{vc1} or \eqref{gvc1}. In first choice, $w_1$ and $w_2$ after the insertion are also endowed with condition \eqref{vc1},  and we obtain the same inequality as above. In second choice, when $v^*_0$ is equipped with \eqref{gvc1}, the splitting can not be performed, and we cannot obtain the desired inequality.}

\begin{theorem}
Suppose $\tilde\Gamma$ is formed by adding an edge of length $\ell$ between two vertices $v$ and $w$ of $\Gamma$. Type and interaction strength of vertex conditions at both vertices $v$ and $w$ are preserved after adding an extra edge. Let $\lambda_{k_0}(\Gamma)\geq (\pi/\ell)^4$ for some integer $k_0$.
If
\begin{enumerate}
\item vertices $v$ and $w$ are equipped with conditions \eqref{vc1} or
\item vertices $v$ and $w$ are equipped with conditions \eqref{gvc1} or
\item vertices $v$ and $w$ are equipped with conditions \eqref{vc2} or
\item vertices $v$ and $w$ are equipped with conditions \eqref{vc1} and \eqref{gvc1}, respectively, or
\item vertices $v$ and $w$ are equipped with conditions \eqref{vc1} and \eqref{vc2}, respectively, or
\item vertices $v$ and $w$ are equipped with conditions \eqref{gvc1} and \eqref{vc2}, respectively, or
\item vertices $v$ and $w$ are equipped with conditions \eqref{vc2} and \eqref{vc2}, respectively, then
$$
\lambda_k(\tilde \Gamma)\leq \lambda_k(\Gamma),\quad k\geq k_0.
$$

If
\item vertices $v$ and $w$ are equipped with conditions \eqref{vc1} and \eqref{gvc2}, respectively, or
\item vertices $v$ and $w$ are equipped with conditions \eqref{gvc1} and \eqref{gvc2}, respectively, or
\item vertices $v$ and $w$ are equipped with conditions \eqref{gvc2} and \eqref{vc2}, respectively, then
$$
\lambda_k(\tilde \Gamma)\leq \lambda_{k+1}(\Gamma),\quad k+1\geq k_0.
$$
\item Finally, if both vertices $v$ and $w$ of $\Gamma$ are equipped with conditions \eqref{gvc2} then
$$
\lambda_k(\tilde \Gamma)\leq \lambda_{k+2}(\Gamma),\quad k+2\geq k_0.
$$
\end{enumerate}
\end{theorem}
\begin{figure}[H]
    \centering
    \includegraphics[scale=0.88]{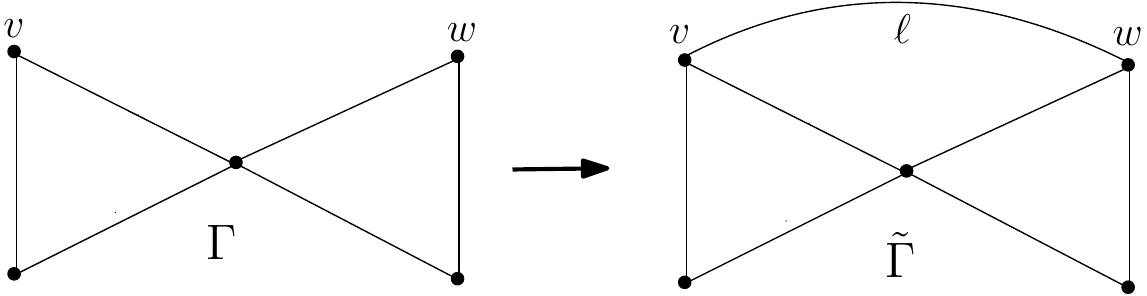}
    \caption{Inserting an edge of length $\ell$ between vertices $v$ and $w$ of $\Gamma$ (Left), we obtain the graph $\tilde\Gamma$ (right).}
\end{figure}
\begin{proof}
\begin{enumerate}
    \item Suppose the vertices $v$ and $w$ of $\Gamma$ are equipped with conditions $\eqref{vc1}$ and consider the operator $\hat{H}=\frac{d^4}{dx^4}$ on a graph $\hat{\Gamma}$ of single edge of length $\ell$ satisfying conditions \eqref{vc2} with $\alpha=0$ on both ends. The lowest eigenvalue $\lambda_1(\hat{\Gamma})$ is equal to zero and $\lambda_2(\hat \Gamma)=\left(\pi/\ell\right)^4$. By gluing one end point of $\hat\Gamma$ with vertex $v$ of $\Gamma$ such that the new vertex also satisfy conditions \eqref{vc1} with the same delta interaction strength as $v$. We denote this new graph by $\check{\Gamma}$. \textcolor{blue}{Theorem} \eqref{pendant} with $r=2$ implies $\lambda_{k+1}(\check \Gamma)\leq \lambda_k(\Gamma)$, for all $k\geq k_0$. Now we glue $w$ and the other endpoint of $\hat \Gamma$ to obtain $\tilde \Gamma$. Part (1) of \textcolor{blue}{Theorem} \eqref{gluing} implies $\lambda_k(\tilde \Gamma)\leq\lambda_{k+1}(\check \Gamma)$ for all $k$. This proves the claim. Parts (2-6) of this theorem can be proved using the same argument.
    \\
    \item[8.] Suppose the vertices $v$ and $w$ of $\Gamma$ are equipped with conditions \eqref{vc1} and \eqref{gvc2}, respectively. As before, we glue one end point of $\hat\Gamma$ and vertex $v$ of $\Gamma$ to obtain $\check\Gamma$. \textcolor{blue}{Theorem} \eqref{pendant} implies $\lambda_{k+1}(\check \Gamma)\leq \lambda_k(\Gamma)$, for all $k\geq k_0$. This can be rewritten as $\lambda_{k+2}(\check \Gamma)\leq \lambda_{k+1}(\Gamma)$, for all $k+1\geq k_0$. Gluing vertex $w$ and the other end point of $\hat\Gamma$ and applying part (2) of \textcolor{blue}{Theorem} \eqref{gluing} we get $\lambda_k(\tilde \Gamma)\leq \lambda_{k+2}(\check \Gamma)$ for all $k$. Proofs of parts $(9)$ and $(10)$ are similar.
    \\
    \item[11.] Suppose both vertices $v$ and $w$ of $\Gamma$ are equipped with conditions \eqref{gvc2} and let us glue $v$ and one end point of $\hat \Gamma$. \textcolor{blue}{Theorem} \eqref{pendant} yields $\lambda_{k+2}(\check \Gamma)\leq \lambda_{k+2}(\Gamma)$, for all $k+2\geq k_0$. Finally, gluing vertex $w$ and the other end point of $\hat\Gamma$ and applying part (2) of \textcolor{blue}{Theorem} \eqref{gluing} we get $\lambda_k(\tilde \Gamma)\leq \lambda_{k+2}(\check \Gamma)$ for all $k$. 
\end{enumerate}
\end{proof}
\section{Bounds on the lowest non-zero eigenvalue}
In this section, we will present some upper and lower bounds on the non-zero lowest eigenvalue. First, we obtained a few upper bounds using trial functions from the domain of the quadratic form. And then, a bound is also provided using the Eulerian cycle approach for a general graph. And for a star graph, a simple lower estimate is expressed involving the maximal length of an edge.
\subsection{Upper bounds}
{ Since $\varphi \equiv 1$ is in the domain of the quadratic form of a graph whose vertices are
equipped with vertex conditions either \eqref{vc1}, \eqref{gvc1}, \eqref{gvc2}, or \eqref{vc2}, therefore, we have a simple upper estimate on the first eigenvalue given by the following inequality.
\begin{equation} \label{ub_1}
    \lambda_1(\Gamma) \leq \frac{\sum\limits_{m=1}^{|V|} \alpha_m}{\mathcal{L}(\Gamma)} .
\end{equation}
Similarly, the function $\varphi(x)=\cos \left(\frac{2 \pi x}{\ell_i} \right)$ on each edge $e_i=[0,\ell_i]$ is also in the domain, so we can use this function to get another upper bound on the lowest eigenvalue.
\begin{equation*}
    \lambda_1(\Gamma) \leq \frac{2}{\mathcal{L}(\Gamma)} \left( \sum\limits_{i=1}^{|E|}  \frac{8 \pi^4}{\ell^3_i}+\sum\limits_{m=1}^{|V|} \alpha_m \right).
\end{equation*}
Since $\varphi(x)$ does not satisfy the vertex condition at $v_m$ unless $\alpha_{m}=0$. Therefore, it can not be an eigenfunction, so we must have strict inequality.
\begin{equation} \label{up2}
    \lambda_1(\Gamma) < \frac{2}{\mathcal{L}(\Gamma)} \left( \sum\limits_{i=1}^{|E|}  \frac{8 \pi^4}{\ell^3_i}+\sum\limits_{m=1}^{|V|} \alpha_m \right).
\end{equation}
Clearly, it can be seen that the upper bound in \eqref{ub_1} is a better estimate than \eqref{up2}. We obtain equality in \eqref{ub_1} if and only if all interaction strengths are zero. 
The function $\varphi(x)\equiv 1 $ is an eigenfunction for any graph $\Gamma$ whose vertices are either equipped with conditions \eqref{vc1}, \eqref{gvc1}, \eqref{gvc2}, or \eqref{vc2} with interaction strengths $\alpha_m=0, m=1,2,\cdots,|V|$, and corresponding eigenvalue is zero. Assume that $\varphi(x)$ is an eigenfunction corresponding to an eigenvalue $\lambda=0$. Then,
\begin{equation*}
    \varphi^{(4)}(x)=0 \implies \varphi(x)=a_i+b_i x+c_i x^2+d_i x^3, \quad x \in e_i,
\end{equation*}
along each edge $e_i$. Since this function also minimizes Rayleigh quotient, this implies that  $\int_\Gamma |\varphi''(x)|^2dx=0$, and thus $\varphi''(x)=0$ implies $c_i=d_i=0$. The eigenfunction is reduced to $\varphi(x)=a_i+b_i x, \ x \in e_i$. When the conditions \eqref{vc2} are imposed at each vertex of $\Gamma$ then the function $\varphi$ satisfies second equation in \eqref{vc2} if and only if $b_i=0$, and the first equation in \eqref{vc2} implies that all $a_i$ are equal and function is constant on whole graph. Furthermore, if a graph is connected, then the eigenfunction $\varphi(x)=1$ is unique corresponding to eigenvalue $\lambda=0$ for graph $\Gamma$ with vertices equipped with either \eqref{gvc1} or \eqref{vc2}. When either \eqref{vc1}, or \eqref{gvc2} conditions with $\alpha_m=0$ are imposed at each vertex of $\Gamma$, we believe that this eigenfunction is not unique and depends on geometry and parametrization of edges.
For example, if we consider a special graph (tree or star graph) such that length of each edge is same and we parameterise each edge in a appropriate direction such that $\varphi$ is continuous on $\Gamma$, and the vertices are equipped with condition \eqref{gvc2} with $\sigma^{v_m}_{x_i}=\sigma^{v_m}_{x_j}$, then the function $\varphi(x)=a+bx$ is an eigenfunction corresponding to eigenvalue $\lambda=0$.
\\

By using the variational principle, We present upper bounds for the spectral gap. From now onward, we assume that the interaction strengths $\alpha_m$ for all conditions at each vertex are zero unless specified. The following upper bound for the conditions \eqref{vc2} with zero potential at all vertices was given in \cite{KM}, but the same result can be obtained for the other three vertex conditions.}  
\begin{proposition}
Let $\Gamma$ be a finite compact connected metric graph equipped with vertex condition \eqref{vc1}, \eqref{gvc1}, (\ref{gvc2}), or \eqref{vc2} at each vertex and let $\alpha_m = 0$ for all $m$. Let $\mathcal{L}(\Gamma)$ be total length of $\Gamma$ and $\ell_i$ be the length of each edge $e_i=[0,\ell_i]$. Then
\begin{equation*}
    \lambda_2(\Gamma) \leq \frac{2}{\mathcal{L}(\Gamma)} \left( \sum\limits_{i=1}^{|E|}  \frac{8 \pi^4}{\ell^3_i}\right).
\end{equation*}
If $\Gamma$ is a bipartite graph, then
\begin{equation*}
    \lambda_2(\Gamma) \leq \left(\frac{\pi}{\mathcal{L}(\Gamma)} \right)^4 \sum\limits_{i=1}^{|E|} \left (\frac{\mathcal{L}(\Gamma)}{\ell_i} \right)^3.
\end{equation*}
\end{proposition}
\begin{proof}
Since $\varphi(x)=\cos(\frac{2 \pi x}{\ell_i}), \ x \in  e_i=[0,\ell_i]$ belongs to the domain of quadratic form; moreover, this function is also orthogonal to constant function. So, using it as test function and with the help of a min-max description of eigenvalues, we arrive at the required result.\\
For second inequality we consider a test function defined by $\varphi(x)=\cos(\frac{ \pi x}{\ell_i})$ on $e_i$, parameterised just like that $\varphi$ is continuous on $\Gamma$, using the same technique, we get required inequality. 
\end{proof}
If the graph $\Gamma$ has certain edges having small lengths, then the bounds from above proposition are reasonably bad. In that case, one can find better bounds by comparing with eigenvalues of other beam operators, as introduced in \textcolor{red}{\cite{KM}}.

\begin{corollary}
If $\Gamma$ is an equilateral graph and length of each edge is $\ell$, then 
\begin{equation*}
    \lambda_2(\Gamma) \leq \left(\frac{2 \pi}{\ell}\right)^4.
\end{equation*}
The number $\left(\frac{2 \pi}{\ell}\right)^4$ is also an eigenvalue, it can either be first non zero eigenvalue, second, or any other.
\end{corollary}
\subsection{Lower bounds}
Generally, the Rayleigh quotient is used to obtain upper estimates on eigenvalues of a metric graph. However, we now provide a lower bound on the lowest non-zero eigenvalue of a star graph with the help of the Rayleigh quotient. The following estimate on the lowest non-zero eigenvalue is derived by following the idea presented in \textcolor{red}{\cite{ZS}}. 
\\
Let $\Gamma$ be an interval of length $\ell$, parameterized by $[0,\ell]$, and each endpoint of $\Gamma$ is equipped with vertex conditions either \eqref{gvc1}, or \eqref{vc2} with zero strengths at all vertices. The functions $\varphi_k(x)=\cos \left( \frac{  k \pi}{\ell} \right)$ are eigenfunctions and  $\lambda_{k}=\left(\frac{ k \pi}{\ell} \right)^4, k=1,2,\cdots$ are corresponding eigenvalues.\\
Consider a star graph $\Gamma$ whose vertices are equipped with vertex conditions \eqref{gvc1} or \eqref{vc2}. The above example shows $\left( \frac{ \pi}{\ell} \right)^4 $ is the first non zero eigenvalue of an edge $e$ of length $\ell$. Assume that $\varphi_1$ is an eigenfunction corresponding to eigenvalue $\lambda_1(\Gamma)$, and now we use the first non-zero eigenvalue of the Bi-Laplacian operator acting on each edge to estimate $\lambda_1(\Gamma)$. Since  $\left( \frac{ \pi}{\ell_i} \right)^4 $ is the first non zero eigenvalue on Bi-Laplacian on each edge $e_i$ of a star graph $\Gamma$, therefore we have 
\begin{equation*}
    \int_{e_i} |\varphi_1''(x)|^2 dx \geq \left( \frac{ \pi}{\ell_i} \right)^4 \int_{e_i} |\varphi_1(x)|^2dx \geq \left( \frac{ \pi}{\ell_{max}} \right)^4 \int_{e_i} |\varphi_1(x)|^2dx,
\end{equation*}
$$  \frac{\int_\Gamma |\varphi''_1(x)|^2dx}{\int_\Gamma |\varphi_1(x)|^2dx} \geq  \left( \frac{ \pi}{\ell_{max}} \right)^4, $$
and now, by the min-max description of eigenvalues, we obtain the following estimate.
\begin{equation*}
    \left( \frac{ \pi}{\ell_{max}} \right)^4 \leq \lambda_1(\Gamma).
\end{equation*}\\

Now we present a lower bound on the spectral gap in terms of the total length of the graph. For standard Laplacian, in \cite{F}  L. Friedlander uses the symmetrization technique, which involves the co-area formula, and showed that the spectral gap is minimum for an interval among all the graphs of the same total length. To the best of our knowledge, we believe the co-area formula does not hold for higher-order derivatives and hence, the symmetrization technique can't be easily generalized for higher-order differential operators. On the other hand, Nicaise in \cite{N} and Kurasov and Naboko in \cite{KN}  provided the same estimates for the spectral gap of a given graph in terms of the spectral gap for a single interval (of the same length) using the Eulerian cycle approach. This approach allows deforming the original graph into various steps leading to graphs with smaller and smaller spectral gaps. We will use this second approach for fourth-order operators to obtain a lower bound on the spectral gap.   Moreover, in \cite{KKK}, it was shown that the lowest eigenvalue of Schr\"odinger operator on compact metric graphs with $\delta$-conditions is bounded below by the lowest eigenvalue of Laplace operator on an interval of the same length with the following robin conditions at the two ends  
$$\varphi'(0)=I_{+} \varphi(0), \quad \varphi'(\mathcal{L})=I_{-} \varphi(\mathcal{L}). $$
Where $I_{+}$ and $I_{-}$ are the total positive and total negative strengths. The proof was based on the Eulerian cycle approach presented in \cite{KN}, in which the following two bounds were established for standard Laplacian on a metric graph $\Gamma$,
$$\lambda_2(\Gamma ) \geq \left(\frac{\pi}{\mathcal{L}} \right)^2.$$
If the degree of each vertex is even, then the above estimate can be improved as,
$$\lambda_2(\Gamma ) \geq 4\left(\frac{\pi}{\mathcal{L}} \right)^2.$$
Similar bounds were proved in \cite{KM} for the lowest non-zero eigenvalue of the beam operator on a metric graph equipped with zero strength condition \eqref{vc2} with zero strengths.Using the same approach a similar estimate can be obtained for lowest non-zero eigenvalue of $\delta$ beam operator on metric graphs equipped with conditions \eqref{gvc2} with $\sigma_{x_i}^{v_m}=\sigma_{x_j}^{v_m}$.\\
We recall vertex conditions \eqref{gvc2} such that
 \begin{equation}\label{eightspecial}
     \sigma_{x_i}^{v_m}=\sigma_{x_j}^{v_m},\quad x_i,x_j\in v_m\quad \mbox{and}\quad  \sigma_{x_i}^{v_m}\in\mathbb{C}, \end{equation}
can be written as
\begin{equation}\label{specialgvc2} \begin{cases} \varphi(x_i)=\varphi(x_j)\equiv\varphi(v_m),\quad x_i,x_j\in v_m,\\
\partial\varphi(x_i)=\partial\varphi(x_j),
\\
\sum_{x_j\in v_m}\partial^2\varphi(x_j)=0,\\
\sum_{x_j\in v_m}\partial^3\varphi(x_j)=-\alpha_m\varphi(v_m).
\end{cases}\end{equation}

\begin{theorem}\label{thmeulerian}
Consider the operator $H_\Gamma=\frac{d^4}{dx^4}$ on a finite, compact and connected metric graph $\Gamma$ of length $\mathcal{L}$ equipped with conditions \eqref{specialgvc2} at all vertices. Also assume that the strengths of delta interactions, $\alpha_m$, at all vertices has the same sign. 
Then, the lowest non-zero eigenvalue $\lambda_0(H_\Gamma)$ is bounded from below by the lowest non-zero eigenvalue of the operator $\hat{H}=\frac{d^4}{dx^4}$ on the loop $\mathcal{C}_{2\mathcal{L}}^v$, which is parameterized by the interval $[0, 2\mathcal{L}]$, of total length $2\mathcal{L}$ and single vertex $v=\{0,2\mathcal{L}\}$ satisfying the conditions

\begin{equation*}\label{intervalconditions}
\begin{cases}\varphi(0)=\varphi(2\mathcal{L}), \\
\varphi'(0)=\varphi'(2\mathcal{L}), \\
\varphi''(0)=\varphi''(2\mathcal{L}),\\
\varphi'''(0)-\varphi'''(2\mathcal{L})=-2\alpha\varphi(0).
\end{cases}
\end{equation*}
\end{theorem}
\begin{proof}
Let $|E|$ and $|V|$, respectively, be the total number of edges and vertices of the metric graph $\Gamma$ and denote by $E_n$ its $n^{th}$ edge between vertices $v_m$ and $v_m'$. We construct a graph $\Gamma_2$ from $\Gamma$ by attaching a new edge $E_n'$ of the same length as $E_n$ between vertices $v_m$ and $v_m'$ for all $n=1,\cdots,|E|$. Consider the operator $H_{\Gamma_2}=\frac{d^4}{dx^4}$ on $\Gamma_2$ subject to the same vertex conditions, \eqref{specialgvc2}, but with strengths of delta interactions equal to $2\alpha_m$ in all vertices. Without loss of generality we can assume that $\alpha_m\geq 0$, $m=1,\cdots,|V|$. We denote by $\alpha:=\sum_{m=1}^M\alpha_m$ the total delta interaction strength.

It is easy to observe that each eigenvalue of $H_\Gamma$ is also an eigenvalue of $H_{\Gamma_2}$. The converse, however, is not true. Indeed: let $\varphi$ be an eigenfunction of $H_\Gamma$ corresponding to eigenvalue $\lambda(\Gamma)$. The eigenfunction $\varphi$ can be extended to $\Gamma_2$ by letting it assume the same values on the new edges $E_n'$ as $E_n$. This newly constructed function $\varphi_2$ on $\Gamma_2$ satisfy the eigenvalue equation on $\Gamma_2$ for the same eigenvalue and the vertex conditions. This, in particular, implies
\begin{equation}\label{gamma2}
    \lambda_0(\Gamma)\geq \lambda_0({\Gamma_2}).
\end{equation}

Every vertex in $\Gamma_2$ has an even degree and therefore there exists a closed Eulerian trail $\mathcal{T}_{2\mathcal{L}}$, of length $2\mathcal{L}$, that traverses each edge precisely once. We can obtain this closed trail by cutting through certain vertices in $\Gamma_2$. We assume that the cutting or splitting of vertices is performed in a way that the nature of vertex conditions and sign of delta interaction strengths are preserved. The closed trail $\mathcal{T}_{2\mathcal{L}}$ has $2|E|$ number of edges  and vertices and each vertex is of degree two. As \eqref{specialgvc2} represent conditions \eqref{gvc2} when 
$$     \sigma_{x_i}^{v_m}=\sigma_{x_j}^{v_m},\quad x_i,x_j\in v_m\quad \mbox{and}\quad  \sigma_{x_i}^{v_m}\in\mathbb{C}\backslash\{0\}$$ and conditions \eqref{vc2} when
$$
     \sigma_{x_i}^{v_m}=\sigma_{x_j}^{v_m}=0,\quad x_i,x_j\in v_m
$$
therefore, \textcolor{blue}{Theorem}  \eqref{splittingtheorem} gives
\begin{equation}\label{trail}
\lambda_0({\Gamma_2})\geq \lambda_0({\mathcal{T}_{2\mathcal{L}}}).
\end{equation}

Let $\varphi_0$ be the eigenfunction corresponding to the eigenvalue $\lambda_0(H_{\mathcal{T}_{2\mathcal{L}}})$. There exists a vertex $v_{\min}$ of $\mathcal{T}_{2\mathcal{L}}$ such that $|\varphi_0(v_{\min})|\leq |\varphi_0(v)|$. Here, $v$ represents an arbitrary vertex of $\mathcal{T}_{2\mathcal{L}}$. Let $\tilde{\mathcal{T}}_{2\mathcal{L}}$ is obtained from $\mathcal{T}_{2\mathcal{L}}$ by assuming zero delta interaction at all vertices  of $\mathcal{T}_{2\mathcal{L}}$ except at $v_{\min}$ where we assume that the strength of delta interaction is equal to $2\alpha$. The eigenfunction $\varphi_0$ belongs to the domain of the quadratic form of $H_{\tilde{\mathcal{T}}_{2\mathcal{L}}}$ and therefore it can be used to estimate $\lambda_0({\tilde{\mathcal{T}}_{2\mathcal{L}}})$:
\begin{equation}\label{zerodeltatrail}
    \lambda_0({\mathcal{T}_{2\mathcal{L}}})\geq\frac{\int_{\mathcal{T}_{2\mathcal{L}}}|\varphi''(x)|^2  dx + 2\alpha|\varphi(v_{\min})|^2}{\int_{\mathcal{T}_{2\mathcal{L}}}|\varphi(x)|^2 dx}\geq \lambda_0({\tilde{\mathcal{T}}_{2\mathcal{L}}}).
\end{equation}
The domain of the quadratic form associated with the operator $H_{\tilde{\mathcal{T}}_{2\mathcal{L}}}$ contains all the functions from the space $W_2^2(\tilde{\mathcal{T}}_{2\mathcal{L}}\backslash V)$ satisfying the following two conditions at all vertices
$$
\begin{cases}\varphi(x_i)=\varphi(x_j), \\
\partial \varphi(x_i)= \partial\varphi(x_j).
\end{cases}
$$
Number of edges and vertices of $\tilde{\mathcal{T}}_{2\mathcal{L}}$ are $2 |E|$ and each vertex is of degree two. Therefore, we can identify $n^{th}$ edge by the interval $[0,\ell_n]$ in such a way that any two consecutive vertices are given by the set of end points $\{0,0\}$ and $\{\ell_n,\ell_{n+1}\}$. This arrangement implies that the condition $\partial \varphi(x_i)= \partial\varphi(x_j)$ is equivalent to the condition $\varphi'(x_i)= \varphi'(x_j)$.

Consider the operator $\hat{H}$ on the loop $\mathcal{C}^v_{2\mathcal{L}}$ of total length $2\mathcal{L}$ and a single vertex $v$. Identify the loop with the interval $[0, 2 \mathcal{L}]$ and consider the following conditions at the single vertex $v$ on the functions from the domain $W_2^4(\mathcal{C}^v_{2\mathcal{L}}\backslash \{v\})$
$$
\begin{cases}\varphi(0)=\varphi(2\mathcal{L}), \\
\varphi'(0)= \varphi'(2\mathcal{L}),\\
\varphi''(0)=\varphi''(2\mathcal{L}),\\
\varphi'''(0)-\varphi'''(2\mathcal{L})=-2\alpha\varphi(0).
\end{cases}
$$
The quadratic form, and its domain, associated with the operator $\hat{H}$ coincides with the quadratic form, and its domain, associated with the operator $H_{\tilde{\mathcal{T}}_{2\mathcal{L}}}$. Therefore,
\begin{equation}\label{hhat}
\lambda_0({\tilde{\mathcal{T}}_{2\mathcal{L}}})=\lambda_0(\mathcal{C}^v_{2\mathcal{L}}).
\end{equation}
Equation \eqref{hhat} along with inequalities \eqref{gamma2}, \eqref{trail} and \eqref{zerodeltatrail}  imply the result.
\end{proof}

In \textcolor{blue}{Theorem}  \eqref{thmeulerian} if the original graph $\Gamma$ is Eulerian, \textit{i.e.} all vertices are of even degree then we can skip the step of constructing $\Gamma_2$ in the proof and can choose an Eulerian trail $\mathcal{T}_{\mathcal{L}}$ on the original graph $\Gamma$ of length $\mathcal{L}$. The lowest non-zero eigenvalue of $H_\Gamma$ is then bounded from below by the lowest non-zero eigenvalue of the operator $\hat{H}$ on the loop $\mathcal{C}_{\mathcal{L}}^v$ of length $\mathcal{L}$ and a single vertex $v$.
\begin{corollary}
     Let degrees of all vertices of $\Gamma$ are even, and all other assumptions of \textcolor{blue}{Theorem}  \eqref{thmeulerian} be satisfied. Then the lowest non-zero eigenvalue of $H_\Gamma$ is bounded from below by the lowest non-zero eigenvalue of the operator $\hat{H}$ on the loop $\mathcal{C}_{\mathcal{L}}^v$ of length $\mathcal{L}$ and a single vertex $v$ equipped with conditions
$$
\begin{cases}\varphi(0)=\varphi(\mathcal{L}), \\
\varphi'(0)= \varphi'(\mathcal{L}),\\
\varphi''(0)=\varphi''(\mathcal{L}),\\
\varphi'''(0)-\varphi'''(\mathcal{L})=-\alpha\varphi(0).
\end{cases}
$$
\end{corollary}
\begin{corollary}\label{pavelcoro}
 If $\alpha_m=0$ for $m=1,\cdots, |V|$ then 
\begin{equation}\label{pavelineq1}
\lambda_0(H_\Gamma)\geq \left(\frac{\pi}{ \mathcal{L}}\right)^4.
\end{equation}
In addition, if $\Gamma$ is Eulerian then
\begin{equation}\label{pavelineq2}
\lambda_0(H_\Gamma)\geq 16 \left(\frac{\pi}{\mathcal{L}}\right)^4.
\end{equation}
\end{corollary}

\noindent \textbf{Remark.} Inequalities \eqref{pavelineq1} and \eqref{pavelineq2} were proved in \textcolor{blue}{Theorem (4)}  of \textcolor{red}{\cite{KM}} for the standard vertex conditions \eqref{vc2} (with $\alpha_m=0$) which require functions first derivative to be equal to zero at vertices. \textcolor{blue}{Corollary}  \eqref{pavelcoro}, on the other hand, does not require that the first derivative of functions from the domain are equal to zero at vertices. Instead, we assume a general condition of continuity of the normal first derivative at the vertices, i.e.,  $$\partial\varphi(x_i)=\partial\varphi(x_j),\quad  x_i, x_j\in v_m.$$
\section{Bounds on higher eigenvalues}
 In this section we discuss bounds on higher eigenvalues values of $\delta$ beam operators on finite compact and connected metric graphs. These bounds are expressed in terms of total length, number of edges $|E|$, number of vertices $|V|$, and first Betti number of a graph. Most of these bounds are for the eigenvalues of a graph with zero interaction strengths at all vertices. 
\begin{definition}
For a finite connected graph $\Gamma$ with $|E|$ and $|V| $ number of edges and vertices, respectively, the Betti number $\beta$ is given by $$\beta=|E|-|V|+1$$
It is the minimum number of edges required to remove from $\Gamma$ to obtain a tree. Equivalently, the Betti number is the total number of independent cycles in $\Gamma$. 
\end{definition}
The following Theorem provides an upper bound on eigenvalues in terms of the length of a single edge. The length of other edges can also be used to obtain a bound. But, we will use a longest edge with a maximum length $\ell_{max}$ to get a better estimate.

\begin{theorem}
Let $\Gamma$ be finite compact metric graph and let $\ell_{max}$ and $\beta$ be the length of a longest edge in $\Gamma$ and the Betti number of the graph, respectively. If each vertex of $\Gamma$ is equipped with vertex conditions either \eqref{gvc1} or \eqref{vc2} with $\alpha_m=0, m=1,2,\cdots,|V|$. Then, 
\begin{equation}\label{ub2}
    \lambda_k(\Gamma) \leq  \left(\frac{(k+\beta) \pi}{\ell_{max}} \right)^4.
\end{equation}
\end{theorem}
\begin{proof}
First, we select the longest edge $e_1$ with length $\ell_{max}$ and add pendants edges to the vertices of $e_1$, producing new boundary vertices and then add more pendant edges to new vertices, repeat this process until we obtain a tree $T$, which has equal number of edges as $\Gamma$. By \textcolor{blue}{Theorem} \eqref{pendant} this  surgical transformations, at each step, lowers all eigenvalues. Second, we join the vertices together to obtain the graph $\Gamma$, and each gluing produces a cycle in the graph. By \textcolor{blue}{Theorem} \eqref{gluing}. The bound is obtained by the following inequalities
$$ \left(\frac{k \pi}{\ell_{max}} \right)^4 =\lambda_k(e_1) \geq \lambda_k(T) \geq \lambda_{k-\beta}(\Gamma).$$
\end{proof}
If the graph is equilateral with $\ell=\ell(e)$ for each edge $e \in E$, then
$$\lambda_k(\Gamma) \leq \left(\frac{(k+\beta)|E| \pi}{\mathcal{L}} \right)^4.$$

\begin{lemma}\label{lemma}
Let $T$ be a finite, compact and connected metric tree and $\Gamma$ be a graph obtained by attaching pendant graphs at some vertices of $T$. Let $e_1$ in $T$ has a maximum length, say $\ell_{max}$. If the vertices of $\Gamma$ are equipped with conditions \eqref{gvc1}  or \eqref{vc2} with $\alpha_m=0$ at all vertices, then
\begin{equation} \label{ub1}
    \lambda_k(\Gamma') \leq \lambda_k(\Gamma) \leq \left (\frac{k \pi}{\ell_{max}} \right)^4.
\end{equation}
\end{lemma}

\begin{proof}
Suppose we want to find an upper estimate for the graph $\Gamma'$, as shown in the figure below.
\begin{figure}[H]
    \centering
    \includegraphics[scale=1.3 ]{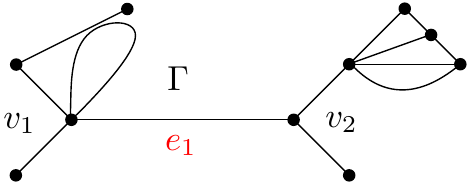}
\end{figure}
The following figure illustrates the idea of the proof, and at each step, we use a surgical toolkit to obtain the inequality.
 \begin{figure}[H]
    \centering
    \includegraphics[scale=1]{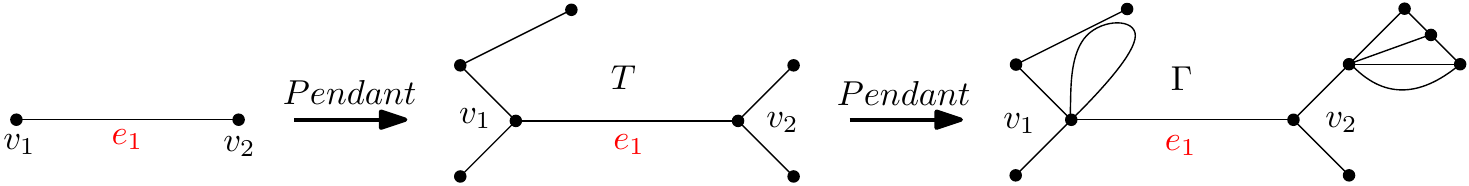}
    \caption{Schematic illustration of the proof of Lemma \ref{lemma}.}
\end{figure}

First, we choose an edge $e_1$ with length $\ell_{max}$ from $T$. Now, we want to reconstruct the graphs $T$ and $\Gamma$ from $e_1$ by using surgery toolkit step by step.\\
At second step, we add pendant edges to vertices $v_1$ and $v_2$ producing new vertices $v^*_i$. Third, we add pendant edges at $v^*_i$, if any. Thus, we have recreated the tree $T$. Finally, we attach pendant graphs at some vertices of $T$ to obtain $\Gamma$. At each step, the surgical transformation lowers all eigenvalues, and we get that the eigenvalues of a graph $\Gamma$ are bounded above by the eigenvalues of a maximal edge $e_1$ of $T$.
\end{proof}

The upper bound \eqref{ub2} in the above theorem can also be used to estimate eigenvalues for the class of graph described in the previous lemma. However, this estimate can be applied to more general graphs compared to \eqref{ub1} but \eqref{ub1} gives better estimate for graphs with large number of cycles. The following two theorems give an estimate of the eigenvalues in terms of the total length of a graph in contrast to length of a single long edge.

\begin{theorem}
Let $\Gamma$ be finite compact metric graph with the total length $\mathcal{L}$ and $\beta$ be the first Betti number of the graph $\Gamma$. If each vertex of $\Gamma$ is equipped with vertex condition either \eqref{gvc1}, or \eqref{vc2} with $\alpha_m=0, m=1,2,\cdots, |V|$. Then, 
\begin{equation}
    \lambda_k(\Gamma) \leq  \left(\frac{(k-3+3|E|+\beta) \pi}{\mathcal{L}} \right)^4.
\end{equation}
\end{theorem}
\begin{proof}
Consider an interval $I$ of total length $\mathcal{L}(\Gamma)$ with end points equipped with condition \eqref{vc2} with strengths zero. Obtain a path graph $P$ by creating the vertex of degree two at interior points of $I$ in accordance with remark \eqref{remarkk} in such a way that the total number of edges in $P$ along with their length coincides with edges in $\Gamma$. Now, we replace the vertex conditions at interior vertices from \eqref{gvc1} to \eqref{vc2} using \textcolor{blue}{Theorem} \eqref{changeofcondition}, call this graph $P'$, and then apply \textcolor{blue}{Theorem}  \eqref{splittingtheorem} to split the path graph $P'$ into disjoint sub-intervals $I_i$ for $i=1,2,\cdots,|E|$. Let $\Gamma'$ be the graph with disconnected components $I_i$, and glue all the intervals $I_i$ together using \textcolor{blue}{Theorem} \eqref{gluing}   to obtain the tree graph $T$, the gluing is performed in such a way that the graph $\Gamma$ can be obtained by pairwise gluing of $\beta$ pairs of vertices of tree $T$.     
\begin{equation*}
    \left(\frac{ k \pi }{\mathcal{L}(\Gamma)} \right)^4=\lambda_k(I)=\lambda_k(P) \geq \lambda_{k-2|E|+2}(P') \geq \lambda_{k-2|E|+2}(\Gamma') \geq \lambda_{k-3|E|+3}(T) \geq \lambda_{k-3|E|+3-\beta}(\Gamma).
\end{equation*}
\end{proof}

The following theorem provides upper and lower estimates on general eigenvalues of a quantum graph $\Gamma$ with vertices equipped with condition \eqref{vc1}. A similar result for  Neumann Laplacian acting on the edges of a metric graph $\Gamma$ was proved in \cite{BKS}. There, they used the rank of the resolvent difference of  Dirichlet and Neumann Laplacian to obtain bounds on the corresponding eigenvalue counting function. In contrast, we have used the interlacing of eigenvalues to obtain the bounds on corresponding counting functions.\\

For any fixed real number $\lambda$, the eigenvalue counting function $N_{\Gamma}(\lambda)$ is defined as the number of eigenvalues of a graph $\Gamma$ smaller than $\lambda$. Since the quantum graph $\Gamma$ is finite compact, and the operator is self-adjoint with a discrete spectrum bounded from below. Therefore the value of function $N_{\Gamma}(\lambda)$ is finite. 
$$ N(\lambda)= \# \{\lambda_i \in \sigma(\Gamma) : \lambda_i \leq \lambda\}. $$

\begin{theorem} \label{eigenvaluecount}
Let $\Gamma$ be finite compact graph of total length $\mathcal{L}$ with $|V|$ number of vertices and $|E|$ number of edges. Let each vertex of $\Gamma$ be equipped with condition \eqref{vc1} with $\alpha_m=0, m=1,2,\cdots,|V|$. Then 
\begin{equation} 
    \left(\frac{\pi}{\mathcal{L}} \right)^4 (k-|V|)^4 \leq \lambda_k(\Gamma) \leq \left(\frac{\pi}{\mathcal{L}} \right)^4 (k+|E|-1)^4.
\end{equation}
\end{theorem}

\begin{proof}
Since the interaction strengths are zero at each vertex, therefore the quadratic form is non-negative, and thus the eigenvalues $\lambda \geq 0$. 
Hence the lower estimate in \eqref{eigenvaluecount}
is interesting only if $k > |V|$.\\
Consider the operator $\frac{d^4}{dx^4}$ acting on a interval of length $\ell$ , and assume that the end points are equipped with condition \eqref{vc1} with $\alpha_m=\infty, m=1,2,\cdots,|V|$. The eigenvalues are $\lambda_k=\left(  \frac{k \pi}{\ell}\right)^4$, for $\lambda \geq 0$ the value of eigencounting  function is $N_{[0,\ell]}(\lambda)=\left[\frac{\sqrt[4] {\lambda }}{\pi} \ell \right].$ Where square brackets mean to take the integer part of the argument. Let $\Gamma_\infty$ be graph obtained from $\Gamma$ by imposing the condition \eqref{vc1} with $\alpha_m=\infty$ at each vertex of $\Gamma$, we shall call them extended condition. Since these conditions imposed at a vertex of degree two or more does not connect the individual function living on the incident edge in any way.  The extended conditions have the effect of disconnecting the vertex of degree $d_v$ into $d_v$ vertices of degree one, so the graph $\Gamma_\infty$ is now decoupled into a set of intervals, and the set of eigenvalues is just the union of eigenvalues of each interval (counting multiplicities). Let  for some $\lambda \in \mathbb{R}$, $N_{\Gamma}(\lambda)$ and $N_{\Gamma_{\infty}}(\lambda)$ denotes the eigenvalue counting functions for the graphs $\Gamma$ and $\Gamma_{\infty}$, respectively. The counting function is given by 
\begin{equation} \label{counting1}
    N_{\Gamma_{\infty}}(\lambda)= \sum\limits_{i=1}^{|E|}N_{[0,\ell_i]}(\Gamma) =\left[\frac{\sqrt[4] {\lambda }}{\pi} \ell_1 \right] +\left[\frac{\sqrt[4] {\lambda }}{\pi} \ell_2 \right] + \cdots+ \left[\frac{\sqrt[4] {\lambda }}{\pi} \ell_{|E|} \right] \leq \left[\frac{\sqrt[4] {\lambda }}{\pi} \mathcal{L} \right].
\end{equation}
 Since for any $a$ and $b$, $ [a]+[b] \leq [a+b]$, therefore taking integer part of the sum of the terms is increased at most by the number of terms minus one as compared to adding integer parts only. As the number of terms in \eqref{counting1} are equal to number of edges $|E|$, thus 
 $$\left[\frac{\sqrt[4] {\lambda }}{\pi} \mathcal{L} \right] =\left[\frac{\sqrt[4] {\lambda }}{\pi} (\ell_1 +\ell_2+\cdots+\ell_{|E|})\right] \leq   \left[\frac{\sqrt[4] {\lambda }}{\pi} \ell_1 \right] +\left[\frac{\sqrt[4] {\lambda }}{\pi} \ell_2 \right] + \cdots+ \left[\frac{\sqrt[4] {\lambda }}{\pi} \ell_{|E|} \right] +|E|-1, $$
 and we have 
 \begin{equation} \label{counting2}
     \left[\frac{\sqrt[4] {\lambda }}{\pi} \mathcal{L} \right] -|E|+1 \leq N_{\Gamma_{\infty}}(\lambda).
 \end{equation}
 The formula \eqref{counting1} and \eqref{counting2}
gives bounds for the eigenvalues of the graph $\Gamma_{\infty}$. 
Let $h$ and $h_{\infty}$ denotes the quadratic form of the graphs $\Gamma$ and $\Gamma_{\infty}$ with domains $D(h)$ and $D(h_{\infty})$, respectively. Since the expression of the quadratic forms $h$ and $h_{\infty}$ are same; moreover, if a function $\varphi$ satisfy the extended conditions at some vertex $v$, then it also satisfies \eqref{vc1} at $v$. Therefore, the domain $D(h_{\infty})$ is a subspace of $D(h)$ and the quadratic forms $h$ and $h_{\infty}$ agree on $D(h_{\infty})$, thus minimizing over a smaller domain results in large eigenvalues. Furthermore, the domain $D(h_{\infty})$ is a co-dimension $|V|$ subspace of $D(h)$. Thus, by the rank $|V|$ nature perturbation, the following interlacing inequalities holds.
 $$\lambda_k(\Gamma) \leq \lambda_k(\Gamma_{\infty}) \leq \lambda_{k+|V|}(\Gamma)$$
  Then the above interlacing inequality implies the following relation between eigenvalue counting functions.
 \begin{equation} \label{cfrelation}
     N_{\Gamma_{\infty}}(\lambda) \leq N_{\Gamma}(\lambda) \leq N_{\Gamma_{\infty}}(\lambda) +|V|.
 \end{equation}
Thus,  
\begin{equation*}
     \left[\frac{\sqrt[4] {\lambda }}{\pi} \mathcal{L}(\Gamma) \right] -|E|+1 \leq N_{\Gamma_{\infty}}(\lambda) \leq N_{\Gamma}(\lambda) \leq N_{\Gamma_{\infty}}(\lambda) +|V| \leq \left[\frac{\sqrt[4] {\lambda }}{\pi} \mathcal{L}(\Gamma) \right]+|V|.
 \end{equation*}
 Setting $\lambda=\left (\frac{\pi k} {\mathcal{L}} \right)^4$ we get 
 $$k-|E|+1 \leq N_{\Gamma}\ \left( \frac{\pi^4 k^4}{\mathcal{L}^4} \right) \leq k+|V|,$$ so
 $$\lambda_{k - |E|+1} \leq \frac{\pi^4}{\mathcal{L}^4} k^4 \leq \lambda_{k+|V|}.$$
 This estimate implies that the multiplicity of the eigenvalues is uniformly bounded by $|E|+|V|$.
 setting $\tilde{k}=k-|E|+1$ we get $\lambda_{\tilde{k}} \leq \frac{\pi^4}{\mathcal{L}^4} (\tilde{k}+|E|-1)^4$ and similarly setting $\tilde{k}=k+|V|$ we get $\frac{\pi^4}{\mathcal{L}^4} (\tilde{k}-|V|)^4 \leq \lambda_{\tilde{k}}$. Since $\tilde{k}$ is a dummy subscript so we can replace it with $k$ and finally, we get,
 $$\left(\frac{\pi}{\mathcal{L}} \right)^4 (k-|V|)^4 \leq \lambda_k(\Gamma) \leq \left(\frac{\pi}{\mathcal{L}} \right)^4 (k+|E|-1)^4.$$
\end{proof}

Using estimates obtained in the previous theorem, we now provide some bounds on the eigenvalues of the same underlying metric graph but equipped with the other three conditions. Here we only proved the result for \eqref{gvc1}, however, this can proved for conditions \eqref{gvc2} and \eqref{vc2} similarly but yields different estimates.  

\begin{theorem} \label{sameconditions}
Let $\Gamma^1$ and $\Gamma^2$ be same underlying metric graphs equipped with conditions \eqref{vc1} and \eqref{gvc1}, respectively, with $\alpha_m=0$ at all vertices. Then
\begin{equation}
\left(\frac{\pi}{\mathcal{L}} \right)^4 (k-|V|)^4 \leq \lambda_k(\Gamma^2) \leq \left(\frac{\pi}{\mathcal{L}} \right)^4 (k+|E|+|V|-1)^4.
\end{equation}
\end{theorem}

\begin{proof}
The proof directly follows from the set of inequalities, as discussed in section $(4)$, and \textcolor{blue}{Theorem}  \eqref{changeofcondition}, along with \textcolor{blue}{Theorem}  \eqref{eigenvaluecount}.
\begin{align*}
    \left(\frac{\pi}{\mathcal{L}} \right)^4 (k-|V|)^4  \leq \lambda_k(\Gamma^1) \leq \lambda_k(\Gamma^2) \leq \lambda_{k+|V|}(\Gamma^1) \leq \left(\frac{\pi}{\mathcal{L}} \right)^4 (k+|E|+|V|-1)^4.
\end{align*}
\end{proof}




\end{document}